%% file: main.tex
\documentclass{article}

\usepackage[utf8]{inputenc}
\usepackage[T1]{fontenc}
\usepackage[english]{babel}
\usepackage{dsfont}
\usepackage{url}            
\usepackage{booktabs}       
\usepackage{amsfonts}       
\usepackage{nicefrac}       
\usepackage{microtype}      
\usepackage{amsmath, amssymb}
\usepackage{fullpage}
\usepackage{authblk}
\usepackage{amsfonts}
\usepackage{bbm}
\usepackage{dsfont}
\usepackage{graphicx}
\usepackage{float}
\usepackage{xcolor}
\usepackage{subcaption}
\usepackage{wrapfig}
\usepackage{tikz}
\usepackage{pst-node}
\usepackage{algorithm}
\usepackage{listings}
\usepackage{fancyvrb}
\fvset{fontsize=\normalsize}
\usepackage[noend]{algpseudocode}
\usetikzlibrary{fit,positioning,bayesnet}
\usetikzlibrary{automata,arrows}
\usetikzlibrary{backgrounds}
\usepackage[acronym,nowarn]{glossaries}
\usepackage{multirow}
\usepackage{multicol}
\usepackage{bm}
\usepackage{upgreek}
\usepackage{natbib}
\usepackage[pagebackref=true]{hyperref}
\renewcommand*{\backrefalt}[4]{%
  \ifcase #1 %
    No citations.%
  \or
    Cited on page #2.%
  \else
    Cited on pages #2.%
  \fi
}

\usepackage[acronym]{glossaries}
\newacronym{mmle}{{{\textsc{\small MMLE}}}}{maximum \textit{marginal} likelihood estimation}
\newacronym{mcmc}{{{\textsc{\small MCMC}}}}{Markov chain Monte Carlo}
\newacronym{pgd}{{{\textsc{\small PGD}}}}{particle gradient descent}
\newacronym{ipla}{{{\textsc{\small IPLA}}}}{interacting particle Langevin algorithm}
\newacronym{ula}{{{\textsc{\small ULA}}}}{unadjusted Langevin algorithm}
\newacronym{kiplmc}{{{\textsc{\small KIPLMC}}}}{kinetic interacting particle Langevin Monte Carlo}
\newacronym{kiplmc1}{{{\textsc{\small KIPLMC1}}}}{kinetic interacting particle Langevin Monte Carlo 1}
\newacronym{kiplmc2}{{{\textsc{\small KIPLMC2}}}}{kinetic interacting particle Langevin Monte Carlo 2}
\newacronym{lvm}{{{\textsc{\small LVM}}}}{latent variable model}
\newacronym{mle}{{{\textsc{\small MLE}}}}{maximum likelihood estimation}
\newacronym{em}{{{\textsc{\small EM}}}}{expectation maximisation}
\newacronym{soul}{{{\textsc{\small SOUL}}}}{Stochastic Optimisation via Unadjusted Langevin}
\newacronym{ips}{{{\textsc{\small IPS}}}}{interacting particle system}
\newacronym{kipld}{{{\textsc{\small KIPLD}}}}{kinetic interacting particle Langevin diffusion}
\newacronym{sde}{SDE}{stochastic differential equation}
\newacronym{mpgd}{{{\textsc{\small MPGD}}}}{momentum particle gradient descent}
\newacronym{klmc}{{{\textsc{\small KLMC}}}}{kinetic Langevin Monte Carlo}
\newacronym{mala}{{{\textsc{\small MALA}}}}{Metropolis adjusted Langevin algorithm}
\newacronym{lmc}{{{\textsc{\small LMC}}}}{Langevin Monte Carlo}
\newacronym{kl}{{{\textsc{\small KL}}}}{Kullback-Leibler}
\newacronym{bnn}{{{\textsc{\small BNN}}}}{Bayesian Neural Network}
\newacronym{klmc1}{{{\textsc{\small KLMC1}}}}{kinetic Langevin Monte Carlo 1}
\newacronym{mpgdnc}{\small MPGDnc}{\small MPGDnc}

\hypersetup{colorlinks,linkcolor={blue},citecolor={bred},urlcolor={blue}}
\setcitestyle{numbers,comma,square}
\usepackage{cleveref}

\usepackage{amsthm}

\newtheorem{proposition}{Proposition}
\newtheorem{theorem}{Theorem}
\newtheorem{remark}{Remark}

\newtheorem{lemma}{Lemma}

\newtheorem{assumptionaux}{}

\newenvironment{assumption}
 {\begin{assumptionaux}\ignorespaces}
 {\end{assumptionaux}}

\newtheorem{assumptionauxG}{}

\newcommand{\argmax}{\operatornamewithlimits{argmax}}

\def\X{{\mathbf X}}

\def\B{{\mathbf B}}
\def\Z{{\mathbf Z}}
\def\V{{\mathbf V}}

\def\bR{{\mathbb R}}

\def\bE{{\mathbb E}}

\def\bN{{\mathbb N}}

\def\f0{{\mathbf 0}}

\def\md{{\mathrm d}}

\definecolor{bred}{rgb}{0.8,0,0}

\title{Kinetic Interacting Particle Langevin Monte Carlo}

\author{Paul Felix Valsecchi Oliva}
\author{O. Deniz Akyildiz}
\affil{Department of Mathematics, Imperial College London}

\affil[]{{\textcolor{blue}{\footnotesize \texttt{\{paul.valsecchi-oliva21, deniz.akyildiz\}@imperial.ac.uk}}}}

\begin{document}
\maketitle

\begin{abstract}
This paper introduces and analyses interacting underdamped Langevin algorithms, termed Kinetic Interacting Particle Langevin Monte Carlo (KIPLMC) methods, for statistical inference in latent variable models. We propose a diffusion process that evolves jointly in the space of parameters and latent variables and show that the stationary distribution of this diffusion concentrates around the maximum marginal likelihood estimate of the parameters. We then provide two explicit discretisations of this diffusion as practical algorithms to estimate parameters of statistical models. For each algorithm, we obtain nonasymptotic rates of convergence in Wasserstein-2 distance for the case where the joint log-likelihood is strongly concave with respect to latent variables and parameters. We achieve accelerated convergence rates clearly demonstrating improvement in dimension dependence. To demonstrate the utility of the introduced methodology, we provide numerical experiments that illustrate the effectiveness of the proposed diffusion for statistical inference. Our setting covers a broad number of applications, including unsupervised learning, statistical inference, and inverse problems.
\end{abstract}
\input{sections/intro}

\input{sections/technicalbackground}

\input{sections/kiplmc}

\input{sections/analysis}

\input{sections/experiments}

\section*{Acknowledgements}
We would like to thank Juan Kuntz and Paula Cordero Encinar for their insightful comments and suggestions.

PVO is supported by the EPSRC through the Modern Statistics and Statistical Machine Learning (StatML)
CDT programme, grant no. EP/S023151/1.

\bibliography{ref}

\newpage
\appendix
\addcontentsline{toc}{section}{Appendices}

\renewcommand{\thetheorem}{\thesection.\arabic{theorem}}
\renewcommand{\thelemma}{\thesection.\arabic{lemma}}
\renewcommand{\theproposition}{\thesection.\arabic{proposition}}
\renewcommand{\theremark}{\thesection.\arabic{remark}}
\renewcommand{\thecorollary}{\thesection.\arabic{corollary}}

\vbox{%
  \vskip 0.25in
  \vskip -\parskip%
    \hsize\textwidth
    \vskip 0.1in
    \centering
    {\LARGE\bf Appendix\par}
 \vskip 0.29in
  \vskip -\parskip
  \vskip 0.09in%
  }
\input{sections/appendix}

\end{document}

%% file: sections/intro.tex
\section{Introduction}\label{sec:intro}

\Glspl*{lvm} are ubiquitous in many areas of statistical science, e.g., complex probabilistic models for text, audio, video and images \citep{blei2003latent, smaragdis2006probabilistic, hoff2002latent}, or inverse problems \citep{glyn2024statistical}. These models capture the underlying structure of the data in terms of low-dimensional variables -- a structure that often exists in the real world \citep{whiteley2022statistical}. However, \glspl*{lvm} often require non-trivial procedures for their estimation as their likelihood is often intractable.

Consider a generic latent variable model as $p_\theta (x,y)$, parameterised by $\theta\in\mathbb{R}^{d_\theta}$, for \textit{fixed}, observed data $y \in \bR^{d_y}$, and latent variables $x \in \bR^{d_x}$. Thus, formally, we see the statistical model as a real-valued mapping $p_\theta(x,y) :\mathbb{R}^{d_x}\times \mathbb{R}^{d_\theta}\to \bR$. The task we are interested in is to estimate the parameter $\theta$ that explains the fixed dataset $y$. Often, this is achieved via the \gls*{mle}. In our setting, due to the presence of latent variables, we aim at finding the \gls*{mmle}, which is termed the \gls*{mmle} problem \citep{dempster1977maximum}. More precisely, our problem takes the form
\begin{align}\label{eq:max}
\bar{\theta}_\star \in \argmax_{\theta\in\mathbb{R}^{d_\theta}} \log p_\theta(y),
\end{align}
where $p_\theta(y) := \int p_\theta(x,y) \mathrm{d}x$ is the \textit{marginal likelihood} (also called the model evidence in Bayesian statistics \citep{bernardo2009bayesian}). It is apparent from \eqref{eq:max} that the problem cannot be solved via optimisation techniques alone for most statistical models.

Classically, this problem is solved with the iterative \gls*{em} algorithm \citep{dempster1977maximum}, which increases the marginal likelihood at every iteration and, under regularity conditions, converges to a stationary point of $\theta \mapsto \log p_\theta(y)$ \citep{wu1983convergence}. At a given parameter estimate $\theta_n$, the \gls*{em} algorithm implements a map $\theta_n \mapsto \argmax_{\theta \in \mathbb{R}^{d_\theta}} \mathbb{E}_{p_{\theta_n} (x|y)}[\log p_{\theta} (x, y)]$, which is guaranteed not to decrease the marginal likelihood at any iteration. This procedure requires an ``E-step'', namely the computation (or estimate) of the expectation w.r.t. the posterior distribution of the latent variables $p_{\theta_n}(x|y)$, and an ``M-step'', which is the maximisation of the expected log-likelihood w.r.t. $\theta$. These steps are intractable in general and can be approximated with a variety of methods, which have been extensively explored. For example, Monte Carlo \gls*{em} \citep{wei1990monte} and stochastic \gls*{em} \citep{celeux1985sem} have been widely studied  \citep{celeux1992stochastic, chan1995monte, sherman1999conditions, booth1999maximizing, cappe1999simulation, diebolt1995stochastic}.

In the most general case, the \gls*{em} algorithm is implemented using \gls*{mcmc} techniques for the E-step \citep{atchade2017}, and numerical optimisation techniques for the M-step \citep{meng1993maximum,liu1994ecme,lange1995gradient}. With the popularity of \gls*{ula} in Bayesian statistics and machine learning \citep{durmus2017nonasymptotic,durmus2019high,dalalyan2017theoretical}, novel variants of \gls*{em} algorithms use unadjusted chains. Most notably, \citep{de2021efficient} studied an algorithm termed \gls*{soul}, which performs unadjusted Langevin steps for the E-step and stochastic gradient ascent for the M-step, building on the ideas of \citep{atchade2017}. In the \gls*{mcmc} based \gls*{em} setting, the bias incurred \emph{by using finitely many} \gls*{mcmc} steps complicates the theoretical analysis and requires stringent conditions for step-sizes.

An alternative approach was developed in \citep{kuntz2023particle} where the authors proposed an interacting particle system, consisting of $N$ particles to replace sequential \gls*{mcmc} steps. This framework, termed \gls*{pgd}, is based on the idea of joint optimisation and sampling to avoid a double-loop structure as in \gls*{mcmc}-based approaches, which results in an efficient algorithmic framework as well as one that is amenable to theoretical analysis \citep{caprio2024error}. Inspired by the approach in \citep{kuntz2023particle}, a closely related interacting particle system was proposed in \citep{akyildiz2025interacting}, where a scaled noise is injected in the $\theta$-dimension. The authors termed this method \gls*{ipla} and proved error bounds for the algorithm. This seemingly small modification is significant, making the algorithm an instance of a Langevin diffusion (an observation we build on in this paper), streamlining the theoretical analysis and opening the door to a wide variety of extensions with theoretical guarantees. The \gls*{ipla} methodology has been already utilised in other contexts, see, \citep{johnston2024taming} for superlinear extensions and \citep{encinar2025proximal} for a set of proximal methods based on \gls*{ipla}, and \citep{akyildiz2024multiscale} for relations between \gls*{ipla} and multiscale methods. \\

\noindent\textbf{Contributions.} In this paper, we develop optimisation methods for \gls*{mmle} based on underdamped Langevin samplers. Our contributions can be summarised as follows: \\

\noindent\textbf{(C1)} We propose the \gls*{kipld}, a diffusion process to optimise the marginal likelihood in latent variable models. We prove that the stationary measure of this diffusion concentrates around the \gls*{mmle} (Propositions~\ref{prop:stationary} and \ref{prop:mmleerror}) and the \gls*{kipld} converges to this measure exponentially fast (Proposition~\ref{prop:kipld_conv}). \\
\glsunset{kiplmc1}
\glsunset{kiplmc2}

\noindent\textbf{(C2)} We then develop our first \gls*{kiplmc} method, termed \gls*{kiplmc1} -- an exponential integrator discretisation of \gls*{kipld}. We prove discretisation error bounds for our method showing convergence rates in both time and step-size (Theorem~\ref{thm:KPLMC1}), closing the problem of nonasymptotic analysis of this method in the strongly convex case. In particular, we show that \gls*{kiplmc1} attains accelerated rates of convergence compared to its overdamped counterparts \gls*{ipla} \citep{akyildiz2025interacting} and \gls*{pgd} \citep{caprio2024error}. While these overdamped diffusion-based \gls*{mmle} methods achieve an $\varepsilon$ error in Wasserstein-2 distance in $\widetilde{\mathcal{O}}(d_x \varepsilon^{-2})$ steps, we prove that \gls*{kiplmc1} attains $\varepsilon$ error in $\widetilde{\mathcal{O}}(\sqrt{d_x} \varepsilon^{-1})$ steps.\\

\noindent\textbf{(C3)} We then introduce a splitting-based explicit discretisation scheme within this setting, which leads to a novel \gls*{mmle} algorithm, termed \gls*{kiplmc2}. We study the nonasymptotic behaviour of this method (Theorem~\ref{thm:KIPLMC2}), showing that the algorithm obtains an error $\varepsilon$ in Wasserstein-2 distance in $\widetilde{\mathcal{O}}(d_\theta\sqrt{d_x}\,\varepsilon^{-3})$ steps.\\

\noindent\textbf{(C4)} Finally, we provide solid empirical evidence of the performance of our methods in a variety of settings, including synthetic and realistic models. In the convex setting, we validate the $\mathcal{O}(N^{-1/2})$ concentration onto the maximiser $\bar{\theta}_\star$  and demonstrate the increased stability of the momentum-based methods. In the non-convex setting, we show the acceleration of the proposed methods, as well as the importance of the choice of numerical integrator, validating our methodological contribution. \\

\noindent\textbf{Relations to existing work.} The \gls*{kipld} is similar to the continuous time system on which the \gls*{mpgd} algorithm by \citep{lim2023momentum} is based (see Eq.~(21) in \citep{lim2023momentum} where we take $\eta_\theta = \eta_x = 1$ and $\gamma_x = \gamma_\theta = \gamma$). In line with the work done by \citep{akyildiz2025interacting}, we inject the $\theta$-dynamics with an appropriately scaled Brownian motion, i.e., $\sqrt{2\gamma/N} \md \B^0_t$. The advantage of studying and implementing this version is due to the fact that the new system can be shown to be an example of a standard underdamped Langevin diffusion, which enables us to show non-asymptotic bounds in a streamlined manner. Furthermore, our discretisations differ from the \gls*{mpgd} routine: the \gls*{kiplmc1} algorithm does not use a gradient correction step, whilst \gls*{kiplmc2} is based on a different integrator. In Sec.~\ref{sec:experiments}, we show that these algorithms obtain similar stability, while having theoretical guarantees. 

The paper is structured as follows. In Section~\ref{sec:background}, we present the background work on underdamped Langevin diffusions for sampling and interacting particle systems for the \gls*{mmle} problem. Following this introduction, in Section~\ref{sec:algorithms}, we present a new diffusion targeting the solution of the \gls*{mmle} problem and present the associated algorithms \gls*{kiplmc1} and \gls*{kiplmc2}. Section~\ref{sec:nonasymp} presents the nonasymptotic analysis of \gls*{kiplmc} methods, providing a clear evidence of improved theoretical bounds due to the acceleration. Finally, in Section~\ref{sec:experiments}, we present a variety of experiments to support our theoretical findings.

\subsubsection*{Notation}
Let $\mathcal{P}(\mathbb{R}^d)$, for $d\geq 1$, denote the set of probability measures on $(\mathbb{R}^d,\mathcal{B}(\mathbb{R}^d))$, where $\mathcal{B}(\mathbb{R}^d)$ is the Borel $\sigma$-algebra. We write $\langle\cdot,\cdot\rangle$ and $\|\cdot\|$ for the Euclidean inner product and norm on $\mathbb{R}^d$. We also use $[N]=\{1,\dots,N\}$ and $\bN$ for the positive integers. For $p\geq 1$, define the Wasserstein-$p$ distance by
\[
W_p(\pi,\nu) = \left(\inf_{\Gamma\in \mathbf{T}(\pi, \nu)} \int_{\mathbb{R}^d\times\mathbb{R}^d} \|x-y\|^p\,\mathrm{d}\Gamma(x,y)\right)^{1/p},
\]
where $\mathbf{T}(\pi,\nu)$ denotes the set of couplings of $\pi$ and $\nu$, i.e., probability measures $\Gamma$ on $\mathbb{R}^d\times\mathbb{R}^d$ with marginals $\pi$ and $\nu$. We also define $\widetilde{\mathcal{O}}$ as the big $\mathcal{O}$ notation, up to logarithmic orders.

%% file: sections/technicalbackground.tex
\section{Technical Background}\label{sec:background}

Before presenting our proposed diffusion and algorithms, we introduce some concepts that will be useful to us in this paper. We first introduce the kinetic Langevin diffusion and its associated algorithms, which are the building blocks of our proposed methods. Then, we present the \gls*{ipla} algorithm, which is the closest method to our proposed algorithms and serves as a starting point for our developments.

\subsection{Kinetic Langevin Monte Carlo}
In recent years, the Langevin diffusion and its associated algorithms have been widely used in statistics and machine learning to sample from complex distributions. The overdamped Langevin diffusion is given by the \gls*{sde}
\begin{equation}\label{eq:LangevinDiffusion}
\mathrm{d}\Z_t = -\nabla U(\Z_t) \mathrm{d}t + \sqrt{2}\mathrm{d}\B_t,
\end{equation}
where $U:\bR^d \to \bR$ is a potential function and $(\B_t)_{t\geq 0}$ is a Brownian motion. Under certain regularity conditions, this system is known to be invariant w.r.t. the measure $\pi(\md z) \propto \exp(-U(z))\md z$ \citep{Pavliotis_2014}. Under the strongly convex and Lipschitz-gradient setting, a standard Euler-Maruyama discretisation of this diffusion, called \gls*{ula}, is known to attain at most $\mathcal{O}(\varepsilon)$ error in Wasserstein-2 distance to the stationary measure in $\widetilde{\mathcal{O}}(d \varepsilon^{-2})$ steps \citep{dalalyan2017theoretical, durmus2019high}.

An alternative to the \textit{overdamped} Langevin diffusions given in \eqref{eq:LangevinDiffusion}, is another class of diffusions called \textit{underdamped} Langevin diffusions \citep{Pavliotis_2014}. This class of diffusions is akin to second-order differential equations and defined over position and momentum variables. In particular, for a $d$-dimensional target measure $\pi(\md z) \propto \exp(-U(z))\md z$, the underdamped (kinetic) Langevin diffusion is given by the \gls*{sde}
\begin{equation}\label{eq:underdamped}
\begin{aligned}
\mathrm{d}\Z_t &= \V_t \mathrm{d}t  \\
\mathrm{d}\V_t &= - \gamma \V_t \mathrm{d}t - \nabla_z U(\Z_t)\mathrm{d}t + \sqrt{2\gamma}\mathrm{d}\B_t, 
\end{aligned}
\end{equation}
where $\gamma > 0$ is called the friction coefficient and $(\B_t)_{t \geq 0}$ is a Brownian motion. Under certain regularity conditions \citep{Pavliotis_2014}, this system is known to be invariant w.r.t. an extended stationary measure of the form
\begin{equation}\label{eq:underdamped_invariant}
\bar{\pi}(\md z, \md v) \propto \exp\left(-U(z) - \frac{1}{2}\|v\|^2\right) \md z \md v.
\end{equation}
This means that we can recover the samples from our target measure $\pi(\md z) \propto \exp(-U(z))\md z$ by sampling from \eqref{eq:underdamped_invariant} and then marginalising out the velocity variable.

The structure in \eqref{eq:underdamped} allows for smoother sample paths and faster convergence to the stationary measure compared to the overdamped case \citep{cheng2018underdamped}. This has motivated the development of algorithms based on discretisations of the underdamped Langevin diffusion, termed \gls*{klmc} \citep{dalalyan2018sampling}. These methods have been shown to achieve improved convergence rates compared to their overdamped counterparts, making them a competitive alternative for sampling from complex distributions. In particular, under the strongly convex and Lipschitz-gradient setting, \gls*{klmc} is known to attain at most $\mathcal{O}(\varepsilon)$ error in Wasserstein-2 distance to the stationary measure in $\widetilde{\mathcal{O}}(\sqrt{d} \varepsilon^{-1})$ steps \citep{dalalyan2018sampling}. In this paper, we build on the framework of \gls*{klmc} methods to develop algorithms for the \gls*{mmle} problem.

\subsection{MMLE via Interacting Particle Langevin Algorithm}
Consider the \gls*{mmle} problem given in \eqref{eq:max}, that is, to find the parameter $\theta$ that maximises the marginal likelihood $p_\theta(y)$. As summarised in the introduction, this problem is often solved with the \gls*{em} algorithm, typically implemented with a double-loop structure consisting of \gls*{mcmc} techniques for the E-step and numerical optimisation techniques for the M-step \citep{atchade2017,de2021efficient}.
\glsreset{pgd}
\glsreset{ipla}

In contrast to this approach, \citep{kuntz2023particle} propose an interacting particle system, the \gls*{pgd} algorithm, to approximate the gradient of the $\theta$-dynamics by running independent particles to integrate out latent variables. Inspired by this, \citep{akyildiz2025interacting} attempt to solve the \gls*{mmle} problem with a similar system, termed the \gls*{ipla}, which is a modification of \gls*{pgd} where $\theta$-dynamics contain a carefully scaled noise. This approach produces a system that is akin to the \gls*{ula} and makes the theoretical analysis streamlined using the analysis produced for \gls*{ula} \citep{dalalyan2017theoretical, durmus2019high}. The proposed algorithm is a discretisation of a system of interacting Langevin \glspl*{sde} which evolves in $\bR^{d_\theta} \times \bR^{N d_x}$ where $N$ distinct particles are retained for latent variables. More precisely, \gls*{ipla} recursions are based on the \gls*{sde}:
\begin{align}
    \mathrm{d}\bm{\theta}_t &= -\frac{1}{N} \sum_{i=1}^N \nabla_\theta U(\bm{\theta}_t, \X_t^{i}) \mathrm{d}t + \sqrt{\frac{2}{N}}\mathrm{d}\B_t^0, \label{eq:ipla1}\\
    \mathrm{d}\X_t^{i} &= -\nabla_x U(\bm{\theta}_t, \X^i_t)\mathrm{d}t + \sqrt{2}\mathrm{d}\B_t^i, \label{eq:ipla2}
\end{align}
for $i \in [N]$, where $(\B_t^0)_{t\geq 0}$ is a Brownian motion evolving on $\bR^{d_\theta}$ and $(\B_t^i)_{t\geq 0}$ for $i \in [N]$ are Brownian motions evolving on $\mathbb{R}^{d_x}$. In this case, \citep{akyildiz2025interacting} observe that the $\theta$-marginal of the stationary measure of this system takes the form $\pi_\Theta(\md \theta) \propto \exp(N \log p_\theta(y)) \md \theta$ which concentrates on the \gls*{mmle} $\bar{\theta}_\star$ as $N$ grows, where $N$ plays the role of inverse temperature \citep{hwang1980laplace}. The authors then identify the convergence rate to the joint stationary measure, as well as an error bound for the discretisation error, from which an error can be determined between the $\theta$-iterate of the algorithm and the \gls*{mmle}.

%% file: sections/kiplmc.tex
\section{Kinetic Interacting Particle Langevin Monte Carlo}\label{sec:algorithms}

To develop our methodology, we start with the following diffusion process:
\small
\begin{align}
\mathrm{d} \bm{\theta}_t &= \V^\theta_t \mathrm{d}t \nonumber \\
\mathrm{d} \X^i_t &= \V^{x_i}_t \mathrm{d}t, \tag{\text{\small\textsc{KIPLD}}}\label{eq:KIPLD}\\
\mathrm{d} \V^{\theta}_t &= -\gamma \V^{\theta}_t \mathrm{d}t - \frac{1}{N}\sum_{i=1}^N \nabla_\theta U(\bm{\theta}_t, \X^{i}_t) \mathrm{d} t + \sqrt{\frac{2\gamma}{N}}\mathrm{d}\B^0_t \nonumber\\
\mathrm{d} \V^{x_i}_t &= -\gamma \V_t^{x_i} \mathrm{d}t - \nabla_x U(\bm{\theta}_t, \X_t^i) \mathrm{d}t + \sqrt{2\gamma} \mathrm{d}\B_t^i, \nonumber
\end{align}
\normalsize
for $ i\in [N]$, where $\{(\B_t^i)_{t\geq 0}\}_{i\in [N]}$ is a family of $\bR^{d_x}$-valued Brownian motions and $(\B_t^0)_{t\geq 0}$ is an $\mathbb{R}^{d_\theta}$-valued Brownian motion. This diffusion has the property that its $\theta$-marginal at stationarity concentrates around the global maximisers of $\log p_\theta(y)$ (see Section~\ref{sec:nonasymp}).

Next, we introduce two numerical integrators for the \gls*{kipld}: firstly, we consider an Exponential Integrator, as discussed in \citep{dalalyan2018sampling}; secondly, a splitting scheme is applied, as described in \citep{monmarche2021high}.

\subsection{Exponential Integrator (KIPLMC1)}\label{subsec:kiplmc1}

In order to derive the exponential integrator discretisation of \gls*{kipld}, following \citep{dalalyan2018sampling}, we begin by defining the functions
\begin{align*}
\psi_0^t &= e^{-\gamma t},\\
\psi_1^t &= \int_0^t e^{-\gamma s} \md s = \frac{1}{\gamma}(1-e^{-\gamma t}),\\
\psi_2^t &= \int_0^t \psi_1^s \md s
= \frac{1}{\gamma^2}(e^{-\gamma t}-1) + \frac{t}{\gamma}.
\end{align*}
The exponential integrator is obtained by freezing the nonlinear drift terms $\nabla_\theta U(\theta_t,X_t^i)$ and $\nabla_x U(\theta_t,X_t^i)$ at the beginning of each time step and solving exactly the resulting linear stochastic differential equation.

More precisely, let $t_n=n\eta$. Over the interval $[t_n,t_{n+1}]$ we approximate the dynamics \eqref{eq:KIPLD} by freezing the gradients at $(\theta_n,X_n^i)$,
\[
F_n^\theta = \frac{1}{N}\sum_{i=1}^N \nabla_\theta U(\theta_n,X_n^i),
\quad
F_n^{x_i} = \nabla_x U(\theta_n,X_n^i),
\]
which yields the linear system
\begin{align*}
\mathrm{d}\bm{\theta}_t &= \V_t^\theta \mathrm{d}t, \\
\mathrm{d}\X_t^i &= \V_t^{x_i} \mathrm{d}t,\\
\mathrm{d}\V_t^\theta &= -\gamma \V_t^\theta \mathrm{d}t - F_n^\theta \mathrm{d}t
+ \sqrt{\frac{2\gamma}{N}}\mathrm{d}\B_t^0,\\
\mathrm{d}\V_t^{x_i} &= -\gamma \V_t^{x_i} \mathrm{d}t - F_n^{x_i} \mathrm{d}t
+ \sqrt{2\gamma}\mathrm{d}\B_t^i,
\end{align*}
for $i\in[N]$, with initial condition
\[
(\bm{\theta}_{t_n},\X_{t_n}^i,\V_{t_n}^\theta,\V_{t_n}^{x_i})
=
(\theta_n,X_n^i,V_n^\theta,V_n^{x_i}).
\]
Since the system is linear with constant coefficients, it can be solved explicitly. Using an integrating factor for the velocity equation gives, for $s\in[0,\eta]$,
\small
\begin{align*}
V_{t_n+s}^\theta
&= \psi_0^s V_n^\theta - \psi_1^s F_n^\theta
+ \sqrt{\frac{2\gamma}{N}}
\int_0^s \psi_0^{\,s-\tau}\,\mathrm{d}\B_{t_n+\tau}^0,\\
V_{t_n+s}^{x_i}
&= \psi_0^s V_n^{x_i} - \psi_1^s F_n^{x_i}
+ \sqrt{2\gamma}
\int_0^s \psi_0^{\,s-\tau}\,\mathrm{d}\B_{t_n+\tau}^i.
\end{align*}
\normalsize
Integrating once more yields the updates for the positions
\small
\begin{align*}
\theta_{t_n+s}
&= \theta_n + \psi_1^s V_n^\theta - \psi_2^s F_n^\theta
+ \sqrt{\frac{2\gamma}{N}}
\int_0^s \psi_1^{\,s-\tau}\,\mathrm{d}\B_{t_n+\tau}^0,\\
X_{t_n+s}^i
&= X_n^i + \psi_1^s V_n^{x_i} - \psi_2^s F_n^{x_i}
+ \sqrt{2\gamma}
\int_0^s \psi_1^{\,s-\tau}\,\mathrm{d}\B_{t_n+\tau}^i.
\end{align*}
\normalsize
Evaluating these expressions at $s=\eta$ gives the one-step update of the discretised process. The stochastic integrals appearing above are jointly Gaussian random variables. For each $i\in\{0,\ldots,N\}$ we define
\begin{align*}
\varepsilon_n^i = \int_0^\eta \psi_0^{\,\eta-s}\,\mathrm{d}B_{t_n+s}^i, \quad \varepsilon_n^{i,\prime} = \int_0^\eta \psi_1^{\,\eta-s}\,\mathrm{d}B_{t_n+s}^i.
\end{align*}
\begin{algorithm}[b!]
\caption{\gls*{kiplmc1}}\label{alg:KIPLMC1}
\begin{algorithmic}[1]
\Require $\gamma, \eta > 0$, $N, M \in \bN$
\State Compute the covariance matrix $C$ given in \eqref{eq:covariancemat} and its Cholesky decomposition $C = H H^\top$.
\State Draw $\theta_0, X^i_0, V^\theta_0, V^{x_i}_0$, for $i\in [N]$
    \For{$n = 0:M-1$}
    \State Draw: $\xi_n^0,\xi_n^{0,\prime}\in\mathbb{R}^{d_\theta}$ and $\xi_n^i, \xi_n^{i,\prime}\in\mathbb{R}^{d_x}$, for $i=1,\dots, N$, from standard Gaussians and define
    \begin{align*}
    \varepsilon_n^{i} &= H_{11} \xi_n^{i}, &
    \varepsilon_n^{i,\prime} = H_{21} \xi_n^i + H_{22}  \xi_n^{i,\prime}
\end{align*}
for $i = 0, \ldots, N$.
    \State Compute the average gradient of the potential:
    \begin{align*}
    F_{n}^{\theta} &= \frac{1}{N} \sum_{i=1}^N \nabla_\theta U(\theta_n, X^i_n), & 
    F_{n}^{x_i} = \nabla_x U(\theta_n, X^i_n).
    \end{align*}
    Update the parameters, particles and their momenta:
    \small
    \begin{align}
    \theta_{n+1} &= \theta_n + \psi_1^\eta V^{\theta}_n - {\psi_2^\eta} F_{n}^{\theta} + \sqrt{\frac{2\gamma}{N}} \varepsilon^{0, \prime}_{n} \nonumber \\
    X^i_{n+1} &= X^i_n + \psi_1^\eta V^{x_i}_n - \psi_2^\eta F_{n}^{x_i} + \sqrt{2\gamma}\varepsilon_{n}^{i,\prime} \tag{KIPLMC1} \label{eq:KIPLMC1} \\
    V^{\theta}_{n+1} &= \psi_0^\eta V^{\theta}_n - \psi_1^\eta F_{n}^{\theta} + \sqrt{\frac{2\gamma}{N}} \varepsilon^0_{n} \nonumber \\
    V^{x_i}_{n+1} &= \psi_0^\eta V^{x_i}_n - \psi_1^\eta F_{n}^{x_i} + \sqrt{2\gamma} \varepsilon^i_{n} \nonumber
    \end{align}
    for $i\in [N]$.
    \normalsize
    \EndFor
\State Output: $(\theta_{M}, X^1_M, \ldots, X^N_M)$
\end{algorithmic}
\end{algorithm}
Now note that $(\varepsilon_n^0, \varepsilon_n^{0,\prime}) \in \mathbb{R}^{d_\theta} \times \mathbb{R}^{d_\theta}$ and $(\varepsilon_n^i, \varepsilon_n^{i,\prime}) \in \mathbb{R}^{d_x} \times \mathbb{R}^{d_x}$ for $i=1,\ldots,N$. Consider each coordinate of these pairs:
\begin{align*}
\varepsilon_{n,j}^i &= \int_0^\eta \psi_0^{\,\eta-s}\,\mathrm{d}B_{t_n+s,j}^i,\qquad
\varepsilon_{n,j}^{i,\prime} = \int_0^\eta \psi_1^{\,\eta-s}\,\mathrm{d}B_{t_n+s,j}^i.
\end{align*}
For every $(i,j)$, using It\^o isometry \citep{oksendal2013stochastic}, we can obtain $C = \mathrm{Cov}(\varepsilon_{n,j}^i, \varepsilon_{n,j}^{i,\prime})$ as
\begin{equation}\label{eq:covariancemat}
C =
\int_0^\eta
\begin{pmatrix}
(\psi_0^{t})^2 & \psi_0^t\psi_1^t\\
\psi_0^t\psi_1^t & (\psi_1^{t})^2
\end{pmatrix}
\mathrm{d}t.
\end{equation}
We note that $C\in\mathbb{R}^{2\times2}$ and denote its Cholesky decomposition by $C=H H^\top$. Hence the Gaussian pair $(\varepsilon_{n,j}^i, \varepsilon_{n,j}^{i,\prime})$ can be sampled by setting
\begin{align*}
\begin{pmatrix}\varepsilon_{n,j}^i\\
\varepsilon_{n,j}^{i,\prime}\end{pmatrix}
&= H
\begin{pmatrix}\xi_{n,j}^i\\
\xi_{n,j}^{i,\prime}\end{pmatrix},
\end{align*}
where $\xi_{n,j}^i$ and $\xi_{n,j}^{i,\prime}$ are independent standard Gaussian random variables.

Substituting these expressions for $s = \eta$ yields the \gls*{kiplmc1} scheme. The full algorithm is given in Algorithm~\ref{alg:KIPLMC1} and a more explicit derivation in Appendix~\ref{app:algderiv}.

\subsection{A Splitting Scheme (KIPLMC2)}

We next introduce \gls*{kiplmc2}, an adaptation of the underdamped Langevin sampler introduced by \citep{HOROWITZ1991247} and extensively studied by \citep{monmarche2021high}. This splitting scheme is termed OBABO and it is a second order scheme, only requiring the first order derivatives of $U$. The idea is to split the numerical scheme into simpler subflows that can be solved explicitly. More precisely, in our \gls*{kipld} case, we decompose the flow into a Ornstein-Uhlenbeck step (O), a velocity drift step (B) and a particle update step (A). 

In the OBABO case, over a full time-step $\eta$, the (O) step is applied over a $\eta/2$ half-step, the (B) step is also applied to a $\eta/2$ half-step, (A) is applied as a full time-step $\eta$, followed by a reversed application of (B) and (O) for half-steps. We denote these intermediate time-steps by $n+\frac{1}{4}$, $n+\frac{1}{2}$ and $n+\frac{3}{4}$, though we note that these time-steps do not correspond to simulated time. The order in which these steps are taken is critical \citep{monmarche2021high} and in our case we will limit ourselves to considering the OBABO ordering.

More precisely, the Ornstein-Uhlenbeck half-step (O) flow is given as
\begin{align*}
\mathrm{d} \V_t^\theta &= -\gamma \V_t^\theta \mathrm{d}t + \sqrt{\frac{2\gamma}{N}} \mathrm{d}\B_t^0,\\
\mathrm{d} \V_t^{x_i} &= - \gamma \V_t^{x_i} \mathrm{d} t + \sqrt{2\gamma}\mathrm{d}\B_t^i,
\end{align*}
over a half-step of duration $\eta/2$, for $i\in[N]$, where the parameters $\bm{\theta}_t$ and $\X_t^i$ are frozen and we initialise at $(\V_{n\eta}^\theta, \V_{n\eta}^{x_i})= (V_n^\theta, V_n^{x_i})$. Using the classic analytic solution to the Ornstein-Uhlenbeck equation and the It\^{o} Isometry we obtain,
\begin{equation}
\begin{aligned}
V_{n+\frac{1}{4}}^\theta &= \delta V_n^\theta + \frac{1}{\sqrt{N}} \sqrt{1-\delta^2} \xi_n^0,\\
V_{n+\frac{1}{4}}^{x_i} &= \delta V_n^{x_i} + \sqrt{1-\delta^2} \xi_n^i,
\end{aligned}\tag{O}
\end{equation}
with $\delta=e^{-\eta\gamma/2}$ and i.i.d. standard Gaussians $\xi_n^0$ and $\xi_n^i$, in $\mathbb{R}^{d_\theta}$ and $\mathbb{R}^{d_x}$, respectively.

The momentum drift update (B) updates $V^\theta$ and $V^{x_i}$ by exactly integrating the flow with frozen gradients $N^{-1} \sum_{i=1}^N \nabla_\theta U(\theta, X^i)$ and $\nabla_{x_i}U(\theta, X^i)$, i.e. we consider the flow to be given as
\begin{align*}
\mathrm{d} \V_t^\theta &= - F_n^\theta \mathrm{d}t,\\
\mathrm{d} \V_t^{x_i} &= - F_n^{x_i} \mathrm{d}t,
\end{align*}
\begin{algorithm}[b!]
\caption{KIPLMC2 Algorithm}\label{alg:KIPLMC2}
\begin{algorithmic}[1]
\Require $\gamma, \eta > 0$, $N, M \in \bN$
\State Draw $\theta_0, X^i_0, V^\theta_0, V^{x_i}_0$, for $i\in [N]$
    \For{$n = 0:M-1$}
    \State Draw $\xi_n^0, \xi^{0,\prime}_n$ from a $d_\theta$-dimensional Gaussian and for $i\in [N]$, draw $\xi_n^i, \xi^{i,\prime}_n$ from $d_x$-dimensional Gaussians.

    \State Compute the gradients of the potential:
    \small
    \begin{align*}
        F_n^\theta &= \frac{1}{N} \sum_{i=1}^N \nabla_\theta U(\theta_n, X_n^i),\quad
        F_n^{x_i} = \nabla_x U(\theta_n, X_n^i).
    \end{align*}
    \normalsize
    \State Update the momenta via (O)+(B) steps:
    \small
    \begin{align*}
    V_{n+\frac{1}{2}}^\theta &= \delta V_n^\theta + \sqrt{\frac{1-\delta^2}{N}} \xi_n^0 - \frac{\eta}{2} F_n^\theta,\\
    V_{n+\frac{1}{2}}^{x_i} &= \delta V_n^{x_i} + \sqrt{1-\delta^2} \xi_n^i - \frac{\eta}{2} F_n^{x_i},
    \end{align*}
    \normalsize
    \State Update parameters and particles via (A):
    \small
    \begin{align*}
    \theta_{n+1} &= \theta_n + \eta V_{n+\frac{1}{2}}^\theta,\quad
    X_{n+1}^i = X_n^i + \eta V_{n+\frac{1}{2}}^{x_i},
    \end{align*}
    \normalsize
    for $i\in [N]$.
    \State Re-compute the gradient potentials 
    \small
    \begin{align*}
    F_{n+1}^\theta &= \frac{1}{N} \sum_{i=1}^N \nabla_\theta U(\theta_{n+1}, X_{n+1}^i), &
    F_{n+1}^{x_i} = \nabla_x U(\theta_{n+1}, X_{n+1}^i).
    \end{align*}
    \normalsize
    \State Update the momenta via the (B)+(O) steps:
    \small
    \begin{align*}
    V_{n+1}^\theta &= \delta V_{n+\frac{1}{2}}^\theta - \frac{\delta \eta}{2}F_{n+1}^\theta +\sqrt{\frac{1-\delta^2}{N}} \xi_n^{0,\prime},\\
    V_{n+1}^{x_i} &= \delta V_{n+\frac{1}{2}}^{x_i} - \frac{\delta\eta}{2} F_{n+1}^{x_i} + \sqrt{1-\delta^2} \xi_n^{i,\prime},
    \end{align*}
    \normalsize
    for $i\in[N]$.
    \EndFor
    
\Return $\theta_{M}$
\end{algorithmic}
\end{algorithm}
over a half-step of duration $\eta/2$, where we recall the definitions of $F_n^\theta$ and $F_n^{x_i}$ from Sec.~\ref{subsec:kiplmc1} and initialise at $(V_{n+\frac{1}{4}}^\theta,V_{n+\frac{1}{4}}^{x_i})$, the output of the (O) half-step. This yields the exact update:
\begin{equation}
\begin{aligned}
V_{n+\frac{1}{2}}^\theta &= V_{n+\frac{1}4}^\theta - \frac{\eta}{2} F_n^\theta,\\
V_{n+\frac{1}2}^{x_i} &= V_{n+\frac{1}4}^{x_i} - \frac{\eta}{2} F_n^{x_i}.
\end{aligned}\tag{B}
\end{equation}
Finally, the (A) update considers the partial flow,
\begin{align*}
\mathrm{d} \bm{\theta}_t &= \V_t^\theta \mathrm{d}t,\\
\mathrm{d} \X_t^i &= \V_t^{x_i}\mathrm{d}t,
\end{align*}
over $t\in [t_n, t_{n+1}]$, where we consider $(\V_t^\theta,\V_t^{x_i})$ is frozen at $(V_{n+1/2}^\theta, V_{n+1/2}^{x_i})$ and we initialise at $(\bm{\theta}_{n\eta}, \X_{n\eta}^i)=(\theta_n, X_n^i)$. Hence, the (A) update is given as,
\begin{equation}
\begin{aligned}
\theta_{n+1} &= \theta_n + \eta V_{n+\frac{1}{2}}^\theta,\\
X_{n+1}^i &= X_n^i + \eta V_{n+\frac{1}{2}}^{x_i}.
\end{aligned}\tag{A}
\end{equation}
The algorithm is then completed by repeating steps (B) and (O),
\begin{equation}
\begin{aligned}
V_{n+1}^\theta &= \delta \left(V_{n+\frac{1}{2}}^\theta -\frac{\eta}{2} F_{n+1}^\theta\right)+ \sqrt{\frac{1-\delta^2}{N}} \xi_n^{0,\prime},\\
V_{n+1}^{x_i} &= \delta\left(V_{n+\frac{1}{2}}^{x_i} - \frac{\eta}{2} F_{n+1}^{x_i}\right)+ \sqrt{1-\delta^2} \xi_n^{i,\prime},
\end{aligned}\tag{B+O}
\end{equation}
where the (B) and (O) subflows are each applied over a half-step of duration $\eta/2$, and where $\xi_n^{0,\prime}$ and $\xi_n^{i,\prime}$ are standard Gaussians in $\mathbb{R}^{d_\theta}$ and $\mathbb{R}^{d_x}$, respectively.

The full updates are given in Algorithm~\ref{alg:KIPLMC2} and a full derivation of the \gls*{kiplmc2} algorithm is given in Appendix~\ref{app:algderiv}.

%% file: sections/analysis.tex
\section{Nonasymptotic Analysis}\label{sec:nonasymp}

In this section, we provide the convergence results for both the analytic and numerical schemes in the nonasymptotic regime. In Section~\ref{sec:assumptions}, we lay out our assumptions about our target measure. In Section~\ref{sec:proofstrat}, we introduce the proof strategy for the reader's convenience, which may be useful for extending and proving similar results. In Section~\ref{sec:conv}, we provide the full nonasymptotic bounds -- in particular, in Section~\ref{sec:conv_stationary}, we identify the stationary measure for KIPLD and show an exponential convergence rate to it, as well as, its concentration onto the \gls*{mmle} solution $\bar{\theta}_\star$ as $N$ grows. Following this, in Section~\ref{sec:numerical}, the error bounds for the two algorithms are provided, allowing us to identify the convergence rate of \gls*{kiplmc1} and \gls*{kiplmc2} to the \gls*{mmle} solution $\bar{\theta}_\star$.

\subsection{Assumptions}\label{sec:assumptions}
Our assumptions are generic for the analysis of Langevin diffusions, i.e., we assume strong convexity and $L$-Lipschitz gradients for the potential $U$. These assumptions are akin to the assumptions made in \citep{durmus2017nonasymptotic,dalalyan2017theoretical,dalalyan2019user} for proving the convergence of \gls*{ula}. It is, however, possible to relax them, as can be seen in \citep{zhang2023nonasymptotic,akyildiz2024nonasymptotic}.

We start with the following assumption on the potential $U$.
\begin{assumption}\label{assump:strongconvexity}
Suppose $U\in C^2(\mathbb{R}^{d_\theta+d_x})$. There exists a constant $\mu>0$ such that, for all $z,z'\in \mathbb{R}^{d_\theta+d_x}$,
\begin{equation*}
\langle z-z', \nabla U(z) - \nabla U(z')\rangle \geq \mu \|z-z'\|^2.
\end{equation*}
Since $U\in C^2$, this is equivalent to $\nabla^2 U(z)\succeq \mu I$ for all $z\in \mathbb{R}^{d_\theta+d_x}$, i.e., $U$ is $\mu$-strongly convex.
\end{assumption}

This assumption is equivalent to an assumption of joint strong convexity in $\theta$ and $x$ for $U$, which ensures exponentially fast convergence of the KIPLD to the global minimiser and places a quadratic lower growth bound on $U$.
\begin{assumption}
\label{assump:lipschitz}
There exists a constant $L>0$ such that, for all $z,z'\in \mathbb{R}^{d_\theta+d_x}$,
\begin{equation*}
\|\nabla U(z) - \nabla U(z')\| \leq L \|z-z'\|.
\end{equation*}
\end{assumption}
Together with Assumption~\ref{assump:strongconvexity}, this implies $\mu\leq L$. This assumption is a common assumption in the analysis of stochastic differential equations, as it ensures stability of the diffusion KIPLD, leading to strong solutions, and stable numerical discretisations \cite{kloeden2012numerical}. Observe that this assumption places an upper quadratic growth bound on $U$ or, from the perspective of the probabilistic model, $\log p_\theta(x,y)$ is of quadratic growth in both $\theta$ and $x$. \ref{assump:strongconvexity} and \ref{assump:lipschitz} together ensure that we can apply the Leibniz differentiation rule under integration for the marginal of the probability model $p_\theta(x, y)= e^{-U(\theta,x)}$ \citep[Theorem~16.8]{billingsley1995probability}.

We remark nonetheless that, in the works cited above, the assumptions are on the strong log-concavity of the target measure. In our setting, this corresponds to imposing a structure on the joint statistical model $p_\theta(x, y)$ in $\theta$ and $x$, as our potential is given as $U(\theta,x) = -\log p_\theta(x,y)$. This requires some care on the user side as a statistical model with strongly log-concave posterior $p_\theta(x | y)$ (with fixed $\theta$) may not necessarily satisfy strong log-concavity of $p_\theta(x, y)$ in $(x, \theta)$ jointly. We verify in the experimental section that for some generic models such as Bayesian logistic regression, strict convexity can be shown (which can be improved to strong convexity with regularisation). As noted above, however, these strong assumptions on the statistical model should not be difficult to relax using the standard non-log-concave sampling results -- but this direction is out of scope of the present work which considers the strongly convex case.

\subsection{The proof strategy}\label{sec:proofstrat}
Let $(\theta_n)_{n\geq 0}$ be the sequence of iterates generated by a numerical scheme, for example, \gls*{kiplmc1} or \gls*{kiplmc2}. Let $\pi_\Theta$ denote the $\theta$-marginal of the stationary measure of the diffusion KIPLD, and $\delta_{\bar{\theta}_\star}$ denote the Dirac measure at $\bar{\theta}_\star$. Finally, we denote the law of $\theta_n$ by $\mathcal{L}(\theta_n)$.

In order to proceed, we first note that the optimisation error is given as $\bE[\|\theta_n - \bar{\theta}_\star\|^2]^{1/2} = W_2(\mathcal{L}(\theta_n), \delta_{\bar{\theta}_\star})$. This is due to the fact that the set of couplings between a measure $\nu$ and another measure $\delta_y$ contains a single element $\delta_y \otimes \nu$ \citep[Section~1.4]{santambrogio2015optimal}, thus the infimum in Wasserstein-2 distance is attained by the coupling $\delta_{\bar{\theta}_\star} \otimes \mathcal{L}(\theta_n)$. This implies that $\bE[\|\theta_n - \bar{\theta}_\star\|^2]^{1/2} = W_2(\mathcal{L}(\theta_n), \delta_{\bar{\theta}_\star})$ and consequently using triangle inequality
\small
\begin{align}
\bE[\|\theta_n - \bar{\theta}_\star\|^2]^{1/2} \leq \underbrace{W_2(\mathcal{L}(\theta_n), \pi_\Theta)}_{\textnormal{convergence}}  + \underbrace{W_2(\pi_\Theta, \delta_{\bar{\theta}_\star})}_{\textnormal{concentration}},\label{eq:triangle}
\end{align}
\normalsize
as the Wasserstein distance is a metric. The first term is the convergence of the numerical scheme to the stationary measure $\pi_\Theta$ which will be proved for each scheme separately. The second term in the right-hand side of \eqref{eq:triangle} is the concentration of the stationary measure $\pi_\Theta$ onto the \gls*{mmle} solution $\bar{\theta}_\star$, which is given by Proposition~\ref{prop:mmleerror} below.

\subsection{Convergence bounds}\label{sec:conv}
As outlined above in \eqref{eq:triangle}, we will first show the concentration of the stationary measure $\pi_\Theta$ onto the \gls*{mmle} solution $\bar{\theta}_\star$, then show the convergence of the KIPLD to the stationary measure $\pi_\Theta$. Finally, we will provide the error bounds for the numerical schemes.

\subsubsection{Concentration of the stationary measure}\label{sec:conv_stationary}
This is a non-standard underdamped Langevin diffusion, thus, we need to first identify the stationary measure of the system. Recall that we are only interested in $\theta$-marginal of this stationary measure and we denote it by $\pi_\Theta$. We first have the following proposition.
\begin{proposition} \label{prop:stationary}
Let $\pi_\Theta$ be the $\theta$-marginal of the stationary measure of the KIPLD. Then, we can write its density as
\begin{align}
\pi_\Theta(\theta) \propto \exp( - N \kappa(\theta)).
\end{align}
where $\kappa(\theta) = - \log p_\theta(y)$.
\end{proposition}
\begin{proof} See Appendix~\ref{app:proofstationary}.
\end{proof}
This result shows that the \gls*{kipld} targets the right object: as $N$ grows, $\pi_\Theta$ will concentrate on the minimisers of $\kappa(\theta)$ by a classical result \citep{hwang1980laplace}. This shows that the number of particles $N$ acts as an inverse temperature parameter in the underdamped diffusion. Since the minimiser of $\kappa(\theta)$ is the maximiser of $\log p_\theta(y)$, we can see that the stationary measure of the \gls*{kipld} is concentrating on the \gls*{mmle} solution $\bar{\theta}_\star$. In particular, under the assumption \ref{assump:strongconvexity}, we have the following nonasymptotic concentration result.
\begin{proposition}\label{prop:mmleerror}
Under \ref{assump:strongconvexity}, the function $\theta \mapsto p_\theta(y)$ is $\mu$-strongly log-concave. Furthermore, we have
\begin{equation*}
W_2(\pi_\Theta, \delta_{\bar{\theta}_\star})\leq \sqrt{\frac{d_\theta}{\mu N}},
\end{equation*}
where $\bar{\theta}_\star = \argmax_{\theta\in\mathbb{R}^{d_\theta}} \log p_\theta(y)$,
which is unique.
\end{proposition}

\begin{proof}
See Appendix~\ref{app:proofmmleerror}.
\end{proof}

We note here that this bound is tight, as it is achieved when $\pi_\Theta$ is Gaussian.

\subsubsection{Convergence of the KIPLD to the stationary measure} \label{sec:numerical}
We next demonstrate that the KIPLD converges to its stationary measure $\pi_\Theta$ exponentially fast in the strongly log-concave case.

\begin{proposition}\label{prop:kipld_conv}
Let $(\bm{\theta}_t)_{t\geq 0}$ be the $\theta$-marginal of a solution to the KIPLD, initialised so that
\begin{align*}
\Z_0
:=(\bm\theta_0,N^{-1/2}\X_0^1,\ldots,N^{-1/2}\X_0^N)^\intercal
\sim\nu,
\end{align*}
where $\nu\in\mathcal P(\mathbb R^{d_z})$ has finite second moment and $d_z=d_\theta+Nd_x$.  Independently of $\Z_0$, initialise the mutually independent velocities as
\begin{align*}
\V_0^\theta&\sim\mathcal N(0,N^{-1}\mathbb I_{d_\theta}),&
\V_0^{x_i}\sim\mathcal N(0,\mathbb I_{d_x}),\qquad i\in[N].
\end{align*}
Then, under \ref{assump:strongconvexity} and \ref{assump:lipschitz} and for $\gamma \geq \sqrt{\mu + L}$, for every $t\geq0$ we have
\begin{align*}
W_2(\mathcal{L}(\bm{\theta}_t), \pi_\Theta) \leq \sqrt{2} \exp\left(-\frac{\mu}{\gamma}t\right) \mathbb{E}[\|\Z_0 - \bar{Z}_\star\|^2]^{1/2},
\end{align*}
where $\bar{Z}_\star\sim\widetilde\pi_Z$, the position marginal of the invariant measure $\widetilde\pi$ in Lemma~\ref{lem:scaledstat}, is drawn independently of $\Z_0$.
\end{proposition}
\begin{proof} See Appendix~\ref{app:proofkipld_conv}. \end{proof}

\subsubsection{Nonasymptotic analysis of KIPLMC1}\label{sec:numericalKPLMC1}
In this section we present our first main result, the convergence rate of \gls*{kiplmc1}. 
\begin{theorem}\label{thm:KPLMC1} Let $(\theta_n)_{n\in\bN}$ be the iterates of \gls*{kiplmc1} and suppose that the process is initialised as $Z_0\sim \nu$, where $\nu\in\mathcal{P}(\mathbb{R}^{d_z})$ has bounded second moments, $V^\theta_0\sim \mathcal{N}(0, N^{-1}\mathbb{I}_{d_\theta})$ and $V^{x_i}_0 \sim \mathcal{N}(0, \mathbb{I}_{d_x})$, for $i=1,\dots,N$. Under Assumptions \ref{assump:strongconvexity} and \ref{assump:lipschitz} and with $\gamma\geq \sqrt{\mu+L}$ and $\eta \leq \mu/(4\gamma L)$, we have
\begin{align*}
\bE[\|\theta_n - \bar{\theta}_\star\|^2]^{1/2} &\leq C_1 \left(1-\frac{0.75\mu\eta}{\gamma}\right)^n + \eta C_2 + \sqrt{\frac{C_3}{N}},
\end{align*}
where
\begin{align*}
C_1 &= \sqrt{2} \bE[\|Z_0 - \bar{Z}_\star\|^2]^{1/2},
\qquad C_2 = \sqrt{2}\frac{L}{\mu} \sqrt{\frac{d_\theta + N d_x}{N}}, \qquad
C_3 = \frac{d_\theta}{\mu},
\end{align*}
where $Z_0=(\theta_0, N^{-1/2} X^1_0,\dots, N^{-1/2} X^N_0)^\intercal$, the initialisation step of the \gls*{kiplmc1} where $\bar{Z}_\star\sim\widetilde\pi_Z$ and $\widetilde\pi_Z$ is the position marginal of the invariant measure $\widetilde\pi$ described in Lemma~\ref{lem:scaledstat}.
\end{theorem}
\begin{proof}
See Appendix~\ref{app:proofnumericalKPLMC1}.
\end{proof}

To provide the algorithmic complexity of \gls*{kiplmc1}, we can inspect the bound provided in Theorem~\ref{thm:KPLMC1}. To obtain a complexity result, we first set $N = \mathcal{O}(d_\theta \varepsilon^{-2})$ which ensures the last term in the bound in Theorem~\ref{thm:KPLMC1} is $\mathcal{O}(\varepsilon)$. Next, set $\eta = \mathcal{O}(\varepsilon d_x^{-1/2})$ which ensures that the second term is $\mathcal{O}(\varepsilon)$. It is then easy to see that $n = \widetilde{\mathcal{O}}(\varepsilon^{-1} d_x^{1/2})$ steps are enough to obtain $\bE[\|\theta_n - \bar{\theta}_\star\|^2]^{1/2} \leq \varepsilon$.

\begin{remark}\label{rem:kiplmc1_complexity} This dependency needs to be compared with \gls*{ipla} and \gls*{pgd}. For the former, we note that our dependence to dimension for $N$ is identical to \gls*{ipla} and the dependence of the number of steps is significantly improved, i.e., we require $n = \widetilde{\mathcal{O}}({d_x}^{1/2} \varepsilon^{-1})$ whereas \gls*{ipla} requires $n = \widetilde{\mathcal{O}}({d_x} \varepsilon^{-2})$ steps for the same accuracy. Regarding the latter method, \gls*{pgd}, our dependence on the number of particles is different, as \gls*{pgd} requires $N = \mathcal{O}(d_x \varepsilon^{-2})$ (but this difference is a more general difference stemming from the different analysis techniques, as \gls*{ipla} also has the same dependence as us to the number of particles $N$, see, e.g., Corollary~8 and the subsequent discussion in \citep{caprio2024error}). However, in terms of the number of steps necessary to attain $\varepsilon$-error, we have the same improvement in dimension dependence in $d_x$ and $\varepsilon$ as we have compared to \gls*{ipla}.
\end{remark}

\subsubsection{Nonasymptotic analysis of KIPLMC2}
We present our second main result, showing the convergence rate of \gls*{kiplmc2}.

\begin{theorem}\label{thm:KIPLMC2}

Let $(\theta_n)_{n\in\mathbb{N}}$ be the iterates of \gls*{kiplmc2}, set $d_z=d_\theta+Nd_x$, and define the rescaled initial position
\begin{align*}
Z_0=(\theta_0,N^{-1/2}X_0^1,\ldots,N^{-1/2}X_0^N)^\intercal.
\end{align*}
Suppose that $Z_0\sim\nu$, where $\nu\in\mathcal{P}(\mathbb{R}^{d_z})$ has finite second moment, and, independently of $Z_0$, initialise the mutually independent velocities as
\begin{align*}
V_0^\theta&\sim\mathcal{N}(0,N^{-1}\mathbb{I}_{d_\theta}), &
V_0^{x_i}\sim\mathcal{N}(0,\mathbb{I}_{d_x}),\qquad i\in[N].
\end{align*}
Under Assumptions~\ref{assump:strongconvexity} and~\ref{assump:lipschitz}, if $\gamma\geq2\sqrt{L}$ and $0<\eta\leq\mu/(33\gamma^3)$, then, for every $n\in\mathbb{N}$,
letting $\widetilde\pi_Z$ denote the position marginal of the invariant measure
$\widetilde\pi$ in Lemma~\ref{lem:scaledstat}, we have
\begin{align*}
\bE[\|\theta_n-\bar{\theta}_\star\|^2]^{1/2}
&\lesssim (\sqrt N\,W_2(\nu,\widetilde\pi_Z) +\eta N\sqrt{d_\theta+Nd_x})
\exp\left(-\frac{\mu\eta n}{6\gamma}\right)\\
&\quad+\eta\sqrt{N(d_\theta+Nd_x)}
+\sqrt{\frac{d_\theta}{N}}.
\end{align*}
The implicit constant depends only on $\mu$, $L$, and $\gamma$; in particular,
it is independent of $n$, $N$, $\eta$, $d_\theta$, and $d_x$.
\end{theorem}

\begin{proof}
See Appendix \ref{app:proofnumericalKPLMC2}.
\end{proof}

To obtain complexity, similarly to the analysis provided for \gls*{kiplmc1}, we first set $N\asymp d_\theta\varepsilon^{-2}$ to obtain an $\mathcal{O}(\varepsilon)$ bound for the last term in Theorem~\ref{thm:KIPLMC2}.  With this choice,
\begin{align*}
\sqrt{N(d_\theta+Nd_x)}
\asymp\frac{d_\theta}{\varepsilon}
\sqrt{1+\frac{d_x}{\varepsilon^2}}.
\end{align*}
We therefore choose
\begin{align*}
\eta
&\asymp\min\left\{\frac{\mu}{33\gamma^3},
\frac{\varepsilon}{\sqrt{N(d_\theta+Nd_x)}}\right\}\\
&\asymp\min\left\{\frac{\mu}{33\gamma^3},
\frac{\varepsilon^2}
{d_\theta\sqrt{1+d_x/\varepsilon^2}}\right\},
\end{align*}
which makes the nondecaying discretisation term of order $\varepsilon$.  In the small-$\varepsilon$ regime, this gives $\eta=\mathcal{O}(\varepsilon^3/(d_\theta\sqrt{d_x}))$.  Finally, writing $D_0:=W_2(\nu,\widetilde\pi_Z)$ and taking
\begin{align*}
n\gtrsim\frac{1}{\eta}
\log\left(1+\frac{\sqrt N\,D_0
+\eta N\sqrt{d_\theta+Nd_x}}{\varepsilon}\right)
\end{align*}
controls both exponentially decaying terms.  Hence, assuming that the initial Wasserstein distance grows at most polynomially in $N$ and the dimensions, $n=\widetilde{\mathcal{O}}(d_\theta\sqrt{d_x}\,\varepsilon^{-3})$ steps are sufficient to obtain an accuracy of order $\varepsilon$, where the notation also suppresses factors depending only on $\mu$, $L$, and $\gamma$.

\begin{table}[t]
    \centering
    \def \arraystretch{1}
    \resizebox{.5\linewidth}{!}{\begin{tabular}{c|c c c}
    \hline
         &  step count ($n$) & $N$ & $\eta$\\
        \hline
       KIPLMC1  &  $\widetilde{\mathcal{O}}(\sqrt{d_x}/\varepsilon)$ & $\mathcal{O}(d_\theta/\varepsilon^2)$ & $\mathcal{O}(\varepsilon/\sqrt{d_x})$ \\
       KIPLMC2  &  $\widetilde{\mathcal{O}}(d_\theta\sqrt{d_x}/\varepsilon^3)$ & $\mathcal{O}(d_\theta/\varepsilon^2)$ & $\mathcal{O}(\varepsilon^3/(d_\theta\sqrt{d_x}))$\\
       IPLA & $\widetilde{\mathcal{O}}(d_x/ \varepsilon^{2})$ & $\mathcal{O}(d_\theta/\varepsilon^2)$ & $\mathcal{O}(\varepsilon^2/d_x)$\\
       PGD & $\widetilde{\mathcal{O}}(d_x/\varepsilon^2)$ & $\mathcal{O}(d_x/\varepsilon^2)$ & $\mathcal{O}(\varepsilon^2 /d_x)$\\
       \hline
    \end{tabular}}
    \caption{Comparison of orders required to achieve an error of order $\varepsilon$ in $W_2$, for step count, particle number $N$ and step size $\eta$.}
    \label{tab:err_rate}
\end{table}

A few comments are in order for this complexity analysis for \gls*{kiplmc2}.
\begin{remark}\label{rem:kiplmc2_complexity} In general, to avoid poor dependence on $N$, the error bounds of numerical discretisations must depend on $L$ and $\mu$ with the matching powers. This ensures that the factors cancel out appropriately, preventing any undesirable effects. This is achieved by the analysis of the exponential integrator in \citep{dalalyan2018sampling} with a linear dependence on the condition number $L/\mu$, thus our \gls*{kiplmc1} bound does not have a bad dependence to $N$. This enables us to obtain improved algorithmic complexity as discussed in Remark~\ref{rem:kiplmc1_complexity}. However, for \gls*{kiplmc2}, the OBABO scheme does not have this dependence on the condition number (see, e.g., \citep[Table~2]{monmarche2021high}) under just \ref{assump:strongconvexity} and \ref{assump:lipschitz}. This results in a bound in Theorem~\ref{thm:KIPLMC2} that grows as $N$ grows, however, we note that it is the authors' belief that this dependence may be an artifact of the proof method employed, rather than reflecting a deteriorating accuracy with growing $N$ for \gls*{kiplmc2}. This is observed in the numerical experiments we present where increasing $N$ does not lead to a growing error.
\end{remark}

See Table \ref{tab:err_rate} for a summary of our discussion.

%% file: sections/experiments.tex
\section{Experiments}\label{sec:experiments}
\begin{figure*}[t]
    \centering
    \centerline{\includegraphics[width=\textwidth]{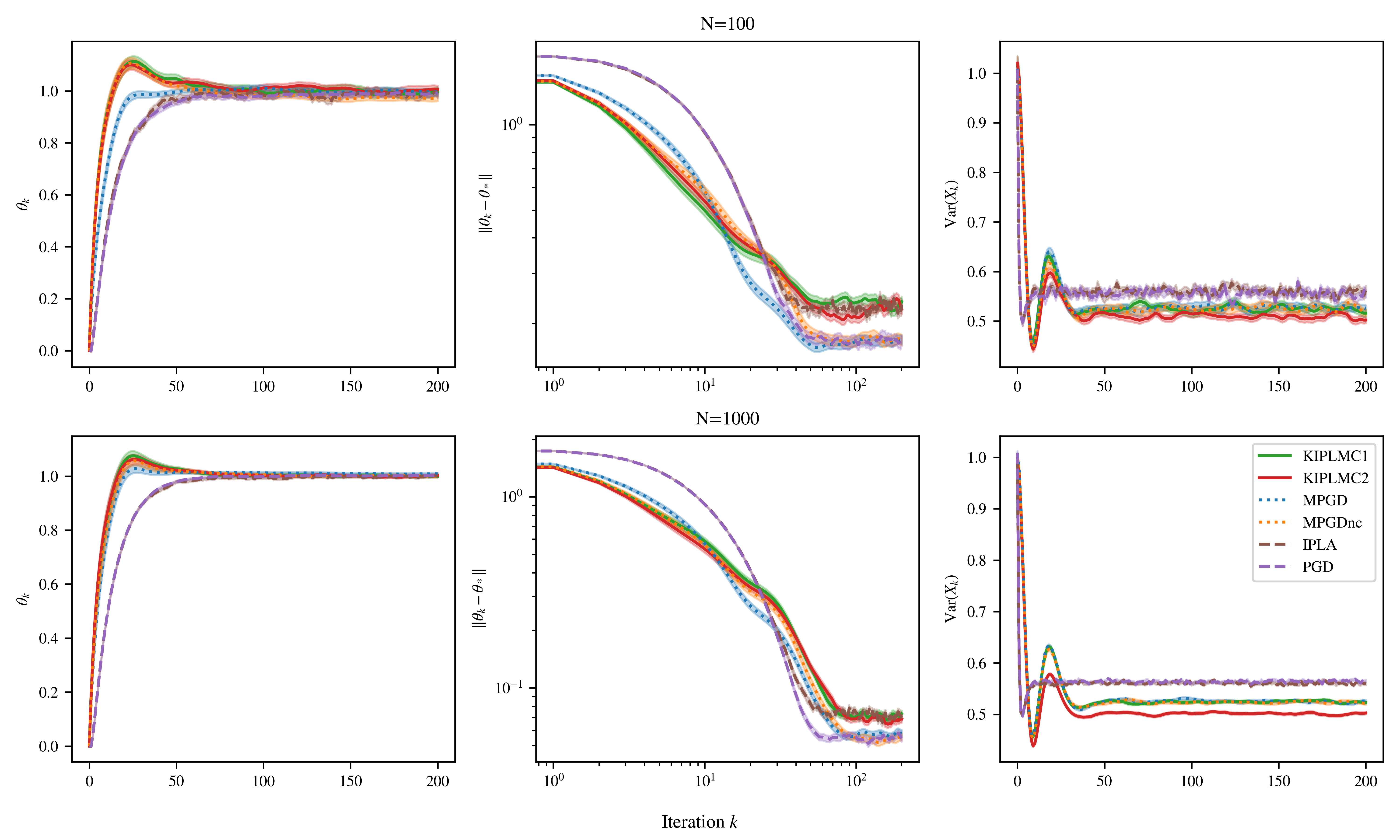}}
    \caption{\textbf{Parameter estimate comparison.} We compare the performance of the \gls*{pgd}, \gls*{ipla}, \gls*{mpgd}, \gls*{kiplmc1}, and \gls*{kiplmc2} algorithms on the synthetic dataset with true $\bar{\theta}_\star = \mathbbm{1}_{d_x}$. We observe the desired convergence of behaviours for larger values of $N$ in the parameter space, mean-square error and the posterior variance of $X_k^i$. In this example $d_x=d_\theta=3$, $d_y=100$ and $\gamma=2.2$. Here we choose $\theta_0, X_0^i, V^\theta_0, V^{x_i}_0\sim\mathcal{N}(0, 0.1)$ and run 100 Monte Carlo simulations.}
    \label{fig:pgdcomparison}
\end{figure*}
In the following section a comparison will be made between the empirical results of \gls*{kiplmc1}, \gls*{kiplmc2} algorithms, as well as, those of \gls*{pgd} from \citep{kuntz2023particle}, \gls*{ipla} from \citep{akyildiz2025interacting} and \gls*{mpgd} from \citep{lim2023momentum}. We note that, in \citep{lim2023momentum}, the \gls*{mpgd} is not solely the discretisation of the KIPLD-like \gls*{sde} using the exponential integrator, but includes a \textit{gradient correction} term, making theoretical analysis and a direct comparison with our methods non-trivial (for a comparison of computational cost see Sec.~\ref{app:compcost}). To this end, we will also compare these methods to MPGDnc, the \gls*{mpgd} algorithm implemented without gradient correction, so that the number of gradient computations is the same as \gls*{kiplmc1}.

In the following sections, we will consider three examples: (i) Bayesian logistic regression on a synthetic dataset, (ii) Bayesian logistic regression on the Wisconsin Cancer dataset, and (iii) a \gls*{bnn} example as a more challenging, non-convex example on the MNIST dataset.

\begin{figure*}[h]
    \centering
    \includegraphics[width=\linewidth]{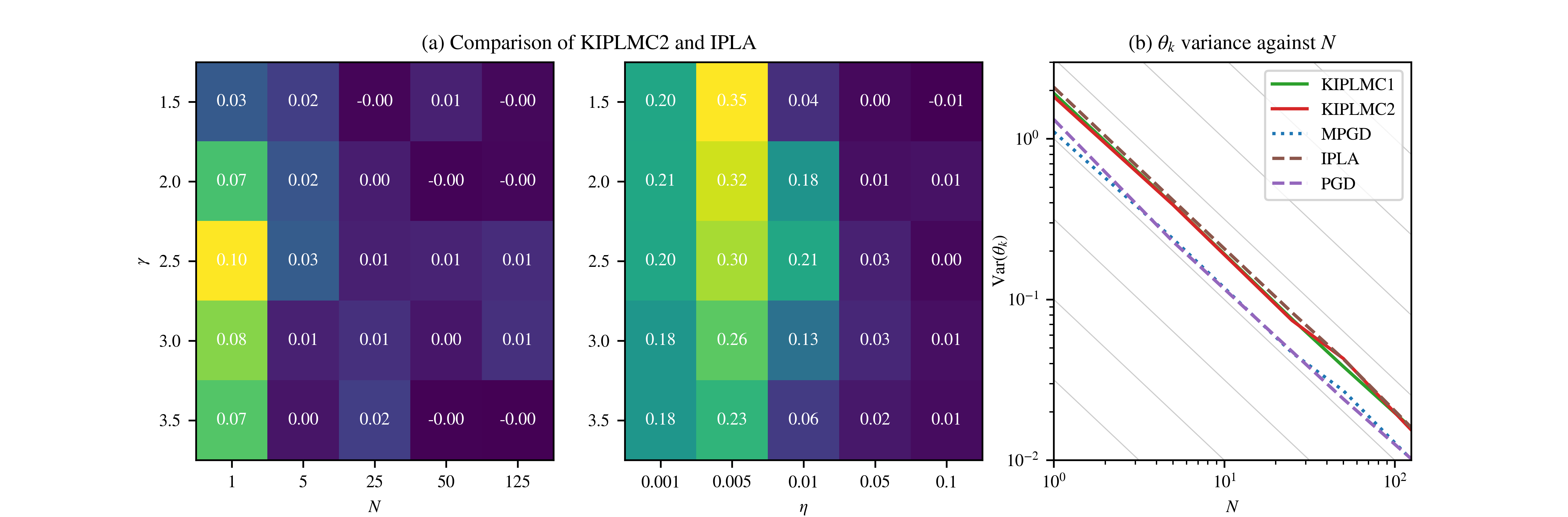}
    \caption{\textbf{Comparison over hyper-parameters.} (a) shows the Area Between the Curves (ABC) values between \gls*{ipla} and \gls*{kiplmc2} for a variety of hyper-parameter combinations. The larger the positive value, the greater the acceleration of \gls*{kiplmc2} with respect to \gls*{ipla} (discussed in C.1). (b) makes a comparison over 100 Monte Carlo simulations of the algorithms' sample variances, using the last 100 steps of each simulation. The grid-lines corresponding to $1/N$ are provided to highlight the rate of convergence. Where not specified otherwise, simulations are run with $M=2000$, $\eta=0.005$, $\gamma=2.2$ and $N=100$.}
    \label{fig:convergence_var}
\end{figure*}

\subsection{Bayesian Logistic Regression on Synthetic Data}\label{sec:toyex}
We follow the experimental setting in \citep{kuntz2023particle} and \citep{akyildiz2025interacting} and start with comparisons between algorithms on a synthetic dataset for which we know the true solution. More precisely, consider the Bayesian logistic regression model
\begin{align*}
p_\theta(x) &= \mathcal{N}(x; \theta, \sigma^2 I_{d_x})\\
p(y|x) &= \prod_{j=1}^{d_y} s(v_j^\intercal x)^{y_j} (1-s(v_j^\intercal x))^{1-y_j}.
\end{align*}
Here $s(u)=e^u/(1+e^u)$ is the logistic function and $v_j$, $j\in[d_y]$ are the set of $d_x$-dimensional covariates with corresponding responses $y_j\in\{0,1\}$. $\sigma$ is given and fixed throughout. We generate a synthetic set of covariates, $\{v_j\}_{j=1}^{d_y}\subset \mathbb{R}^{d_x}$, from which are simulated a synthetic set of observations $y_j|\bar{\theta}_\star, x, v_j$, for fixed $\bar{\theta}_\star$, via a Bernoulli random variable with probability $s(v_j^\intercal x)$. The algorithm is tested on the recovery of this value of $\bar{\theta}_\star$. 

The marginal likelihood is given as,
\begin{align*}
&p_\theta (y)\\
&= \int \left(\prod_{j=1}^{d_y} s(v_j^\intercal x)^{y_j} (1-s(v_j^\intercal x))^{1-y_j} \right) p_\theta(x)\mathrm{d}x.
\end{align*}
From this it easy to see that the gradients of $U$ are given as,
\begin{equation}\label{eq:logreggrad}
\begin{aligned}
 \nabla_\theta U(\theta, x) &= -\frac{x-\theta}{\sigma^2},\\
 \nabla_x U(\theta,x) &= \frac{x-\theta}{\sigma^2} - \sum_{j=1}^{d_y} (y_j - s(v_j^\intercal x))v_j.   
\end{aligned}
\end{equation}

\begin{remark}[On \ref{assump:strongconvexity} and \ref{assump:lipschitz}]\label{rem:convexity}

We will discuss this problem and our assumptions. Writing $z=(\theta,x)$ and $z'=(\theta',x')$, using
$\|(a-b,-a+b)\|\leq2\|(a,b)\|$ and the fact that the logistic function is $1/4$-Lipschitz, we obtain
\begin{equation*}
\begin{aligned}
\|\nabla U(z)-\nabla U(z')\|
&\leq\left(\frac{2}{\sigma^2}+\frac14\sum_{j=1}^{d_y}\|v_j\|^2\right)\|z-z'\|.
\end{aligned}
\end{equation*}
Hence, \ref{assump:lipschitz} is satisfied \citep{akyildiz2025interacting}.

For \ref{assump:strongconvexity}, note that $d_\theta=d_x$ in this model. Set
$q_j(x):=s(v_j^\intercal x)\bigl(1-s(v_j^\intercal x)\bigr)$ and
$Q(x):=\sum_{j=1}^{d_y}q_j(x)v_jv_j^\intercal$. Then
\begin{equation*}
\begin{aligned}
\nabla^2U(z)
&=\frac1{\sigma^2}
\begin{pmatrix}
I_{d_x}&-I_{d_x}\\
-I_{d_x}&I_{d_x}
\end{pmatrix} +\begin{pmatrix}
0&0\\
0&Q(x)
\end{pmatrix}.
\end{aligned}
\end{equation*}
Indeed, for $u,w\in\mathbb R^{d_x}$, $(u, w)^\intercal\nabla^2U(z)(u,w)$ vanishes if and only if $u=w$ and $v_j^\intercal w=0$ for every $j$.
Thus, $U$ is strictly convex provided that
$\operatorname{span}\{v_1,\ldots,v_{d_y}\}=\mathbb R^{d_x}$. Consequently, \ref{assump:strongconvexity} is not satisfied without further regularisation; nevertheless, we believe that this remains a useful illustrative example.
\end{remark}

In Fig.~\ref{fig:pgdcomparison} we can see the difference in the behaviours between the algorithms. Most notably, the \gls*{kiplmc1} and \gls*{kiplmc2} algorithms exhibit comparable levels of acceleration as the \gls*{mpgd} algorithm in the $\theta$-dimension, whilst the \gls*{ipla} and \gls*{pgd} algorithms trail further behind. Interestingly, we observe (more notably for large $N$) that the noise in the $\theta$-dynamics in \gls*{kiplmc1} and \gls*{kiplmc2}, as compared to \gls*{mpgd}, seems to dampen some of the momentum effects when $\gamma$, the friction coefficient is sub-optimal. As $N$ grows the $\theta$ iterates concentrate onto the \gls*{mmle} $\bar{\theta}_\star$ for all algorithms. This behaviour can be seen in more detail in Fig.~\ref{fig:convergence_var} (b), in which we can verify the variance changing with rate $\mathcal{O}(1/N)$. Further, we make a comparison of the performance of the \gls*{ipla} and \gls*{kiplmc2} in Fig.~\ref{fig:convergence_var} (a) with the Area Between the Curves (ABC) metric (see \ref{sec:ABC} for more detail). This metric is a signed, weighted difference between the performances of the two algorithms: the higher the value, the better \gls*{kiplmc2} performs compared to \gls*{ipla}; 0 is when they perform the same. We observe that the acceleration of \gls*{kiplmc2} can be best observed in specific $\gamma$ regimes, which ensure that the momentum is close to critical dampening.

Note that the momentum effects of the \gls*{kiplmc1} and \gls*{kiplmc2} algorithms has been dampened through a specific choice of $\gamma$. In training it was observed that the convergence rate is very sensitive to the choice of $\gamma$, where choices too small, exhibit large momentum effects, and too large, lead to slow convergence (see also \citep{lim2023momentum} for more discussion on the role of the momentum parameter for learning with KLMC). The choice of $\gamma$ here is far from optimal, but allows us to observe the strength of the proposed algorithms.

\begin{figure*}[b!]
    \centering
    \centerline{\includegraphics[width=\linewidth]{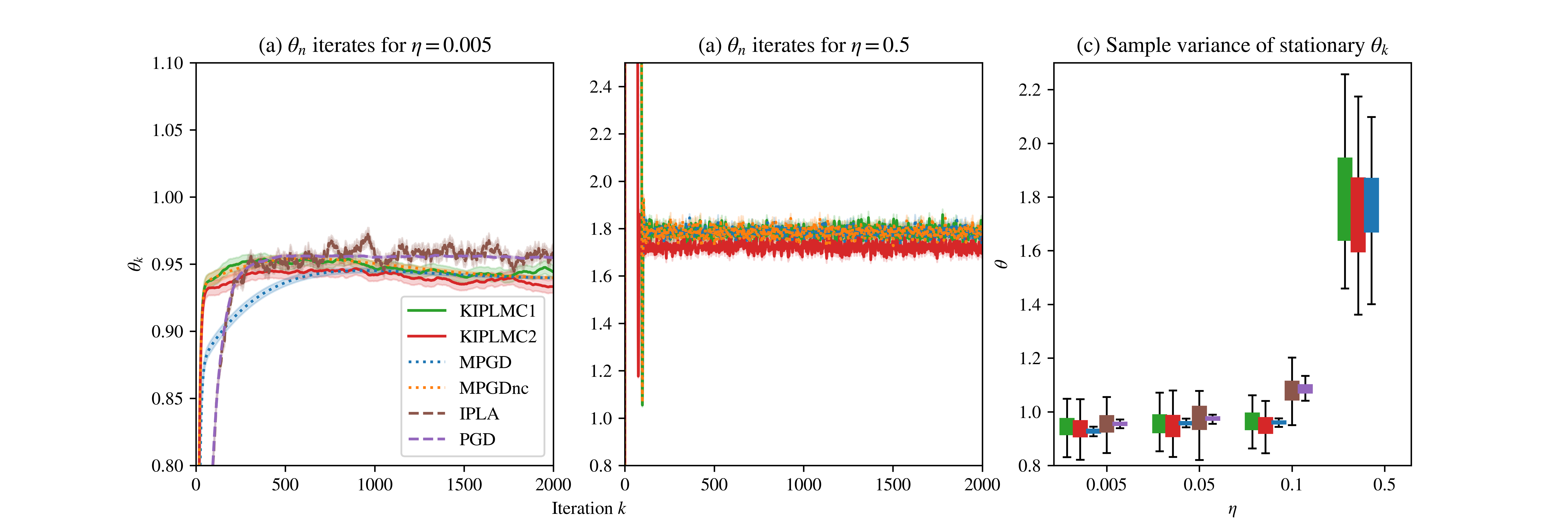}}
    \caption{\textbf{Wisconsin Dataset.} The performance of \gls*{pgd}, \gls*{ipla}, \gls*{mpgd}, \gls*{kiplmc1} and \gls*{kiplmc2} algorithms are compared on a logistic regression experiment for the Wisconsin Cancer Dataset. In (a), we show the behaviour of the $\theta_n$ iterates for a small step-size and (b) shows the behaviour $\theta_n$ iterates for a step-size where \gls*{ipla} and \gls*{pgd} explode. In (c) we compare the distributions of the algorithms over different step-sizes. We set $N=100$ and $\gamma=3.2$.}
    \label{fig:wisconsin}
\end{figure*}

\subsection{Wisconsin Cancer Data}
\label{sec:cancex}

We follow again an experimental procedure that is similar to the one outlined in \citep{kuntz2023particle} and \citep{akyildiz2025interacting} and make comparisons between the algorithms on a more realistic dataset: the Wisconsin Cancer Data. Again, we use the logistic regression LVM model, outlined above. This task is a binary classification, to determine from 9 features gathered from tumours and 683 data points labelled as either benign or malignant. The latent variables correspond to the features extracted from the data. The task in this case is to seek to model the behaviour as accurately as possible through the logistic regression LVM.

For this setup we define our probability model as,
\begin{equation*}
p_\theta(x) = \mathcal{N}(x;\theta \mathbbm{1}_{d_x}, 5 I_{d_x})
\end{equation*}
and the likelihood as,
\begin{equation*}
p(y|x)=\prod_{j=1}^{d_y} s(v_j^\intercal x)^{y_j}(1-s(v_j^\intercal x))^{1-y_j}.
\end{equation*}
Note that the parameter $\theta$ is a scalar in this case, hence, this turns out to be a simplified version of the setup described above and so the discussion in Remark~\ref{rem:convexity} is still valid. As opposed to the previous case with synthetic data, here we consider a real dataset, so no ground-truth $\bar{\theta}^\star$ is available. Hence, there is no comparison to a $\bar{\theta}_\star$, but we can see that different algorithms attain similar estimates of this maximiser.

\begin{figure*}[t!]
    \centering
    \centerline{\includegraphics[width=\linewidth]{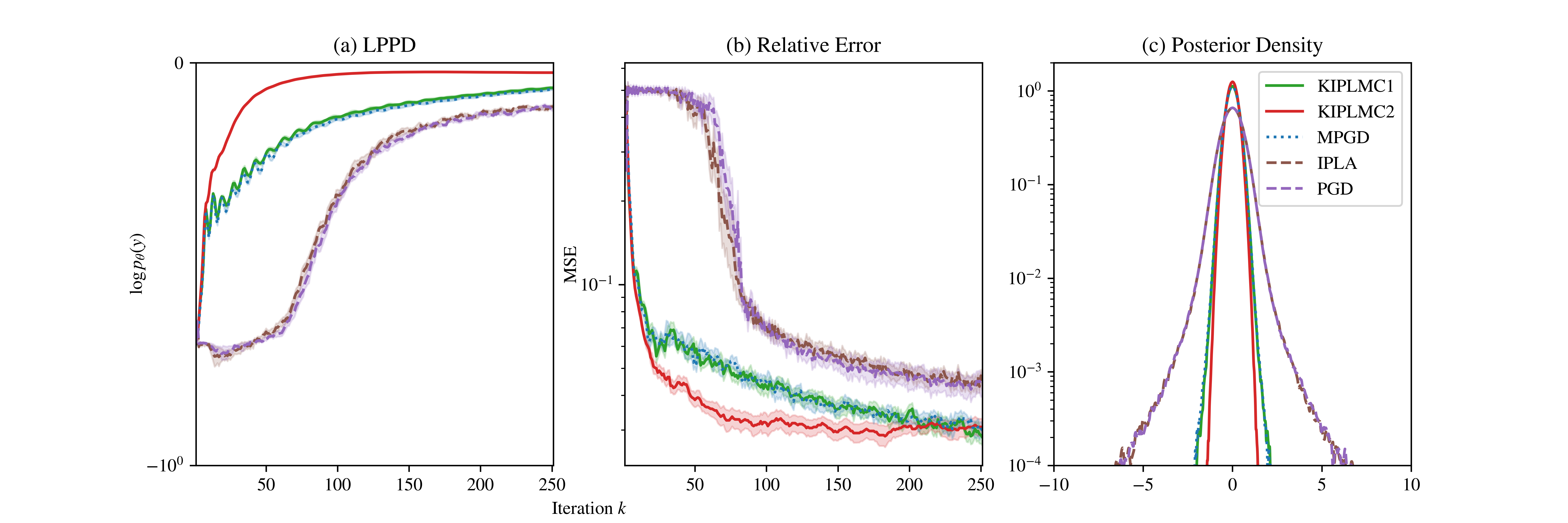}}
    \caption{\textbf{MNIST Dataset.} The performance of \gls*{pgd}, \gls*{ipla}, \gls*{mpgd}, \gls*{kiplmc1} and \gls*{kiplmc2} algorithms are compared on a classification experiment for the MNIST Dataset. (a) shows the Log Pointwise Predictive Density, the average log probability assigned to the correct response. (b) shows the percentage error and (c) the posterior density of the weights $w$. For this experiment $N=100$, $\gamma=2$ and $\eta=0.01$.}
    \label{fig:bnn}
\end{figure*}

Most notably in this experiment, we can observe in Fig.~\ref{fig:wisconsin} the importance of the added stability that the momentum-based algorithms exhibit, with the use of more stable numerical integrators than Euler-Maruyama. In particular, these algorithms displays great stability w.r.t. the choice of the step-size, as well as, the acceleration one would expect. For small step-sizes, the \gls*{pgd} algorithm performs with low variance in all cases where it converges, until it explodes where step-sizes become too large . It is interesting to note however that in the large step-size regime, the momentum-based algorithms exhibit much greater momentum than in the smaller step-sizes. However, again, we note that the \gls*{kiplmc1} and \gls*{kiplmc2} algorithms exhibit more variance in the $\theta$ estimation than the \gls*{pgd} and \gls*{mpgd} algorithms, though it is damped when compared to the \gls*{ipla} algorithm, due to the higher regularity of the solutions to KIPLD. This is typically an advantage when working with convex problems, but this ``sticky'' behaviour might prove detrimental in the non-convex case \citep{gao2022global, akyildiz2025interacting}. The injection of noise into the parameter estimation may help the method to escape local minima \citep{akyildiz2024nonasymptotic}.

\subsection{Bayesian Neural Network}\label{sec:bnn}
Similarly to \citep{kuntz2023particle, lim2023momentum} we also consider a Bayesian Neural Network (BNN) example to perform character classification on the MNIST dataset. In this case our approach is similar to ``cold'' posterior estimation, as our sampling density is proportional to a density raised to the inverse temperature, as in \citep{coldposterior}.

This dataset and model provide a more challenging problem for the \gls*{kiplmc1} and \gls*{kiplmc2} algorithms as the posteriors are known to be multi-modal and hence do not satisfy assumption \ref{assump:strongconvexity}. Indeed, we consider this example for illustrative purposes, showing the relative strength of our proposed methods compared to similar \gls*{ips} methods in a non-convex setting. MNIST contains 70'000 $28\times 28$ grey-scale images $\{f_i\}_{i=1}^{70000} \subset \mathbb{R}^{784}$. However, to avoid issues of high dimensionality, we consider a normalised subset of 1,000 characters containing only images of fours and nines, whose similarity should pose a challenge. Due to this smaller sample, we are also able to avoid the need for sub-sampling.

Following \citep{kuntz2023particle} we employ a two layer BNN with $\tanh$ activation function, softmax output layer and 40 dimensional latent space. i.e. we consider the probability of the data labels $l$ to be conditionally independent, given the features $f$ and the weights $x=(w,v)$, with model,
\begin{equation*}
p(l|f, x) \propto \exp\left(\sum_{j=1}^{40} v_{lj} \tanh\left(\sum_{i=1}^{784} w_{ji} f_i\right)\right),    
\end{equation*}
where $w\in \mathbb{R}^{40\times 784}$, $v\in\mathbb{R}^{2\times 40}$. The model has no bias terms and we place the following Gaussian priors on the weights: $w\sim \mathcal{N}(0, e^{2\alpha}\mathbb{I})$ and $v\sim \mathcal{N}(0, e^{2\beta}\mathbb{I})$. Rather than assigning these priors, we learn the parameters $\theta=(\alpha, \beta)\in\mathbb{R}^2$ jointly with the latent variables $x$. Hence, the model's density is given as
\small
\begin{align*}
p_\theta(x, \{l\}_{i=1}^{1000} |\{f_i\}_{i=1}^{1000})= \mathcal{N}(w; 0,e^{2\alpha}) \mathcal{N}(v; 0, e^{2\beta})\prod_{i=1}^{1000} p(l_i|f_i, x).
\end{align*}
\normalsize
We will employ autograd methods from the PyTorch library to compute the gradients.

In Fig.~\ref{fig:bnn} we make a comparison between the different algorithms over 100 runs, comparing the log predictive point-wise density (LPPD) and relative error. The LPPD is the average log-likelihood assigned to the correct response, whilst the relative error is given as the accuracy of prediction on a test set (see \ref{sec:err} for more detail). Observe the improved behaviour of the momentum-based algorithms in LPPD, as well as relative error. We also note the significant acceleration observed for \gls*{kiplmc2}, highlighting the effect of the distinct numerical integrators. Similarly, in Fig.~\ref{fig:bnn} (c) we can see the more concentrated posterior distributions of the momentum-based algorithms. Thus, it is clear that the \gls*{kiplmc1} and \gls*{kiplmc2} algorithms perform competitively against other particle-based methods, even in very complex problems, in which our assumptions may not be satisfied.

\section{Conclusions}
This paper extends a line of work on interacting particle-, and more generally, diffusion-based algorithms for maximum marginal likelihood estimation. This paper has focused on alternatives to the \glspl*{ips} proposed by \citep{kuntz2023particle, akyildiz2025interacting, lim2023momentum} for the \gls*{mmle} problem by considering an accelerated variant. We have shown that we can leverage the existing literature on underdamped Langevin diffusion for sampling \citep{dalalyan2018sampling, cheng2018underdamped, monmarche2021high, Nesterovacc} to produce two algorithms with greater stability, added smoothness and exponential convergence, which concentrate onto the \gls*{mmle} with quantitative nonasymptotic bounds on the estimation error $\bE[\|\theta_n - \bar{\theta}_\star\|^2]^{1/2}$. In particular, we show that \gls*{kiplmc1} attains an accelerated error rate for error $\varepsilon$ after $\widetilde{\mathcal{O}}(\sqrt{d_x} \varepsilon^{-1})$ steps, compared to the \gls*{ipla} and \gls*{pgd} which require $\widetilde{\mathcal{O}}(d_x \varepsilon^{-2})$ steps. The proposed algorithms both perform well empirically compared to the \gls*{mpgd} and other particle methods. Specifically, the induced noise should improve guarantees in the non-convex setting, as discussed in \citep{akyildiz2025interacting, raginsky2017non,akyildiz2024nonasymptotic}. 

The results presented here hold under assumptions of gradient-Lipschitzness and strong log-concavity, however, it is possible to extend this line of work to the non-convex, super-linear or even proximal cases (for examples of this see, \citep{zhang2023nonasymptotic,akyildiz2024nonasymptotic,johnston2023kinetic, encinar2025proximal}). Another interesting direction is to consider the setting in \citep{fan2023gradient}, using an accelerated algorithm for high-dimensional linear models. Our methods can also be adapted to solving linear inverse problems as done by \citep{akyildiz2022statistical,glyn2024statistical}. Similar ideas can be considered within the setting of Stein variational gradient descent as done in \citep{sharrock2024tuning}. Crucially, the procedure used to obtain the bounds for the algorithms presented here can easily be extended to other cases, creating a powerful theoretical framework for further underdamped Langevin sampling algorithms for the \gls*{mmle} problem.

%% file: sections/appendix.tex
\section{Preliminary results}

\setcounter{theorem}{0}
\renewcommand{\thetheorem}{A.\arabic{theorem}}
\setcounter{lemma}{0}
\renewcommand{\thelemma}{A.\arabic{lemma}}
\setcounter{proposition}{0}
\renewcommand{\theproposition}{A.\arabic{proposition}}
\setcounter{remark}{0}
\renewcommand{\theremark}{A.\arabic{remark}}

In this section, we prove a key result that rewrites the \gls*{kipld} as a single underdamped Langevin diffusion evolving in $\bR^{d_\theta + N d_x}$ with a rescaled potential function and friction coefficient. This will be useful in the analysis of the \gls*{kipld} and its discretisation, as it allows us to leverage existing results on underdamped Langevin diffusions.
\begin{proposition}\label{prop:kipld_standard_ud}
Let $$(\bm{\theta}_t,\X_t^1,\ldots,\X_t^N,\V_t^\theta,\V_t^{x_1},\ldots,\V_t^{x_N})_{t\geq 0}$$ solve \eqref{eq:KIPLD}. Define
\begin{align*}
\widetilde{\bm{\theta}}_t &:= \bm{\theta}_{\sqrt{N}t}, \qquad
\widetilde{\X}_t^i := N^{-1/2}\X^i_{\sqrt{N}t}, \\
\widetilde{\V}_t^\theta &:= \sqrt{N}\V^\theta_{\sqrt{N}t}, \qquad
\widetilde{\V}_t^{x_i} := \V^{x_i}_{\sqrt{N}t},
\end{align*}
for $i\in[N]$, and set
\begin{align*}
\widetilde{\Z}_t &:= (\widetilde{\bm{\theta}}_t,\widetilde{\X}_t^1,\ldots,\widetilde{\X}_t^N)^\intercal \in \bR^{d_\theta+Nd_x}, \\
\widetilde{\V}_t^z &:= (\widetilde{\V}_t^\theta,\widetilde{\V}_t^{x_1},\ldots,\widetilde{\V}_t^{x_N})^\intercal \in \bR^{d_\theta+Nd_x},
\end{align*}
and
\begin{align}
\bar U_N(\theta,x_1,\ldots,x_N) &:= \frac{1}{N}\sum_{i=1}^N U(\theta,\sqrt{N}\,x_i), \label{eq:barU}
\end{align}
and $\widetilde{\gamma} := \sqrt{N}\gamma$. Then there exists a $\bR^{d_\theta+Nd_x}$-valued Brownian motion $(\widetilde{\B}_t)_{t\ge 0}$ such that
\begin{align}
\md \widetilde{\Z}_t &= \widetilde{\V}_t^z\,\md t, \label{eq:simplesde_scaled} \\
\md \widetilde{\V}_t^z &= -\widetilde{\gamma}\,\widetilde{\V}_t^z\,\md t
- N\nabla_z \bar U_N(\widetilde{\Z}_t)\,\md t
+ \sqrt{2\widetilde{\gamma}}\,\md \widetilde{\B}_t. \nonumber
\end{align}
\end{proposition}
\begin{proof}
By definition of the rescaled variables, we have
\begin{align*}
\md\widetilde{\bm\theta}_t = \md \bm\theta_{s(t)} = \V^\theta_{s(t)}\,\md s(t) = \sqrt{N}\,\V^\theta_{\sqrt N t}\,\md t = \widetilde\V^\theta_t\,\md t
\end{align*}
where $s(t):=\sqrt{N}t$. From \eqref{eq:KIPLD}, $\md \X_t^i=\V_t^{x_i}\md t$, so
\begin{align*}
\md \X^i_{s(t)}
= \V^{x_i}_{s(t)}\md s(t)
= \sqrt{N}\,\V^{x_i}_{\sqrt{N}t}\md t.
\end{align*}
Hence
\begin{align*}
\md \widetilde{\X}_t^i
= N^{-1/2}\,\md \X^i_{\sqrt{N}t}
= N^{-1/2}\,\sqrt{N}\,\V^{x_i}_{\sqrt{N}t}\md t
= \widetilde{\V}_t^{x_i}\md t,
\end{align*}
and therefore $\md \widetilde{\Z}_t=\widetilde{\V}_t^z\md t$. Next, define $\widetilde{\B}_t^j:=N^{-1/4}\B_{\sqrt{N}t}^j$ for $j=0,\ldots,N$; then each $\widetilde{\B}_t^j$ is a Brownian motion and
\(
\md \B_{\sqrt{N}t}^j = N^{1/4}\md \widetilde{\B}_t^j.
\)
Using \eqref{eq:KIPLD}, set $s(t):=\sqrt{N}t$. Then
\small
\begin{align*}
&\md \widetilde{\V}_t^\theta
= \sqrt{N}\,\md \V^\theta_{s(t)}\\
&= \sqrt{N}\left[
\left(-\gamma \V^\theta_{s(t)}
- \frac{1}{N}\!\sum_{i=1}^N \nabla_\theta U(\bm{\theta}_{s(t)},\X^i_{s(t)})\right)\md s(t)
+ \sqrt{\frac{2\gamma}{N}}\,\md \B^0_{s(t)}
\right]\\
&= -N\gamma \V^\theta_{\sqrt{N}t}\md t
- \sum_{i=1}^N \nabla_\theta U(\bm{\theta}_{\sqrt{N}t},\X^i_{\sqrt{N}t})\md t  + \sqrt{2\gamma\sqrt{N}}\,\md \widetilde{\B}_t^0\\
&= -\widetilde{\gamma}\,\widetilde{\V}_t^\theta\md t
- N\nabla_\theta \bar{U}_N(\widetilde{\Z}_t)\md t
+ \sqrt{2\widetilde{\gamma}}\,\md \widetilde{\B}_t^0,
\end{align*}
\normalsize
where we used $\md s(t)=\sqrt{N}\md t$ and
\[
N\nabla_\theta \bar{U}_N(\widetilde{\Z}_t)
= \sum_{i=1}^N \nabla_\theta U(\bm{\theta}_{\sqrt{N}t},\X^i_{\sqrt{N}t}).
\]
Fix $i\in[N]$. Using again $s(t)=\sqrt{N}t$ and \eqref{eq:KIPLD},
\small
\begin{align*}
\md \widetilde{V}_t^{x_i}
&= \md \V^{x_i}_{s(t)}\\
&= \left(-\gamma \V^{x_i}_{s(t)}
- \nabla_x U(\bm{\theta}_{s(t)},\X^i_{s(t)})\right)\md s(t)  + \sqrt{2\gamma}\,\md \B^i_{s(t)} \\
&= -\sqrt{N}\gamma \V^{x_i}_{\sqrt{N}t}\md t
- \sqrt{N}\,\nabla_x U(\bm{\theta}_{\sqrt{N}t},\X^i_{\sqrt{N}t})\md t
 + \sqrt{2\gamma\sqrt{N}}\,\md \widetilde{\B}_t^i\\
&= -\widetilde{\gamma}\,\widetilde{\V}_t^{x_i}\md t
- N\nabla_{z_i}\bar{U}_N(\widetilde{\Z}_t)\md t
+ \sqrt{2\widetilde{\gamma}}\,\md \widetilde{\B}_t^i,
\end{align*}
\normalsize
where
\[
N\nabla_{z_i}\bar{U}_N(\widetilde{\Z}_t)
= \sqrt{N}\,\nabla_x U(\bm{\theta}_{\sqrt{N}t},\X^i_{\sqrt{N}t}).
\]
Stacking coordinates and defining
\(
\widetilde{\B}_t
= (\widetilde{\B}_t^0,\widetilde{\B}_t^1,\ldots,\widetilde{\B}_t^N)^\intercal
\)
gives \eqref{eq:simplesde_scaled}.
\end{proof}

Next we show that the stationary measure of the dynamics in \eqref{eq:simplesde_scaled} is given by a Gibbs measure with potential function $N\bar{U}_N$.
\begin{lemma}[Stationary measure for \eqref{eq:simplesde_scaled}]\label{lem:scaledstat}
The measure $\tilde{\pi}$ invariant for the dynamics in \eqref{eq:simplesde_scaled} is given as
\begin{equation}
\tilde{\pi}(\md \tilde{z}, \md \tilde{v}) \propto \exp\left(-N\bar{U}_N(\tilde{z}) -\frac{1}{2}\|\tilde{v}\|^2\right)\md \tilde{z} \md \tilde{v},
\end{equation}
for all $\tilde{z},\tilde{v}\in\mathbb{R}^{d_z}$, where $d_z=d_\theta+Nd_x$.
\end{lemma}

This result follows directly by observing the rescaling given in Proposition~\ref{prop:kipld_standard_ud} and the standard form of the stationary measure for underdamped Langevin diffusions. We now show that $\bar{U}_N$ is $\mu$-strongly convex under \ref{assump:strongconvexity}.

\begin{lemma}[Strong convexity of $\bar{U}_N$]\label{lem:strongconvexity_barU} Under \ref{assump:strongconvexity}, the function $\bar{U}_N: \bR^{d_\theta + Nd_x} \to \bR$ as defined in \eqref{eq:barU} is $\mu$-strongly convex
\begin{align*}
    \langle z-z^\prime, \nabla \bar{U}_N(z) - \nabla \bar{U}_N(z^\prime)\rangle \geq \mu \|z-z^\prime\|^2,
\end{align*}
for all $z, z^\prime \in \bR^{d_\theta + Nd_x}$ and $N \in \bN$.
\end{lemma}
\begin{proof}
Let
\[
z=(z_\theta,z_1,\dots,z_N),\qquad
z'=(z_\theta',z_1',\dots,z_N').
\]
From
\[
\bar U_N(z)=\frac1N\sum_{i=1}^N U(z_\theta,\sqrt N\,z_i),
\]
we have
\begin{align*}
\nabla_{z_\theta}\bar U_N(z) &=\frac1N\sum_{i=1}^N \nabla_\theta U(z_\theta,\sqrt N z_i), &
\nabla_{z_i}\bar U_N(z) = \frac1{\sqrt N}\nabla_x U(z_\theta,\sqrt N z_i).
\end{align*}
Hence
\begin{align*}
\langle& z-z', \nabla \bar U_N(z)-\nabla \bar U_N(z')\rangle =\frac1N\sum_{i=1}^N
\langle z_\theta-z_\theta',\Delta_{\theta,i}\rangle
+\frac1{\sqrt N}\sum_{i=1}^N
\langle z_i-z_i',\Delta_{x,i}\rangle,
\end{align*}
where
\[
\Delta_{\theta,i}:=\nabla_\theta U(z_\theta,\sqrt N z_i)-\nabla_\theta U(z_\theta',\sqrt N z_i'),
\]
\[
\Delta_{x,i}:=\nabla_x U(z_\theta,\sqrt N z_i)-\nabla_x U(z_\theta',\sqrt N z_i').
\]
Now apply Assumption~\ref{assump:strongconvexity} to the two points
\[
(z_\theta,\sqrt N z_i),\qquad (z_\theta',\sqrt N z_i').
\]
For each \(i\),
\begin{align*}
\langle z_\theta-z_\theta',\Delta_{\theta,i}\rangle
&+\sqrt N\,\langle z_i-z_i',\Delta_{x,i}\rangle
\ge \mu\big(\|z_\theta-z_\theta'\|^2+N\|z_i-z_i'\|^2\big).
\end{align*}
Divide by \(N\) and sum over \(i\):
\begin{align*}
\langle z-z', &\nabla \bar U_N(z)- \nabla \bar U_N(z')\rangle \ge
\mu\|z_\theta-z_\theta'\|^2+\mu\sum_{i=1}^N\|z_i-z_i'\|^2
=\mu\|z-z'\|^2.
\end{align*}
So \(\bar U_N\) is \(\mu\)-strongly convex.
\end{proof}

Next, we show that the function $\bar{U}_N$ is $L$-gradient Lipschitz.
\begin{lemma}[Gradient Lipschitzness of $\bar{U}_N$]\label{lem:gradientlipschitz_barU}
Under \ref{assump:lipschitz}, the function $\bar{U}_N: \bR^{d_\theta + Nd_x} \to \bR$ as defined in \eqref{eq:barU} is $L$-gradient Lipschitz
\begin{align*}
    \|\nabla \bar{U}_N(z) - \nabla \bar{U}_N(z^\prime)\| \leq L \|z-z^\prime\|,
\end{align*}
for all $z, z^\prime \in \bR^{d_\theta + Nd_x}$ and $N \in \bN$.
\end{lemma}
\begin{proof}
Let $z = (z_\theta, z_1, \ldots, z_N)$ and $z^\prime = (z_\theta^\prime, z_1^\prime, \ldots, z_N^\prime)$. Then
\begin{align*}
\|\nabla \bar{U}_N(z) - \nabla \bar{U}_N(z^\prime)\|^2 = & \left\| \frac{1}{N} \sum_{i=1}^N  \left(\nabla_\theta U(z_\theta, \sqrt{N}z_i) - \nabla_\theta U(z_\theta^\prime, \sqrt{N}z_i^\prime)\right)\right\|^2 \\
&+ \frac{1}{N} \sum_{i=1}^N \left\|\nabla_x U(z_\theta, \sqrt{N}z_i) - \nabla_x U(z_\theta^\prime, \sqrt{N}z_i^\prime)\right\|^2.
\end{align*}
Using the fact that $\|\cdot\|^2$ is a convex function and Jensen's inequality for the first term, we obtain
\begin{align*}
\|\nabla \bar{U}_N(z) - \nabla \bar{U}_N(z^\prime)\|^2 &\leq\frac{1}{N} \sum_{i=1}^N \left\|\nabla_\theta U(z_\theta, \sqrt{N}z_i) - \nabla_\theta U(z_\theta^\prime, \sqrt{N}z_i^\prime)\right\|^2\\
& + \frac{1}{N} \sum_{i=1}^N \left\|\nabla_x U(z_\theta, \sqrt{N}z_i) - \nabla_x U(z_\theta^\prime, \sqrt{N}z_i^\prime)\right\|^2.
\end{align*}
Using \ref{assump:lipschitz}, we have
\begin{align*}
\|\nabla \bar{U}_N(z) - \nabla \bar{U}_N(z^\prime)\|^2 &\leq\frac{L^2}{N} \sum_{i=1}^N \left(\|z_\theta - z_\theta^\prime\|^2 + N \|z_i - z_i^\prime\|^2\right)\\
&= L^2 \|z - z^\prime\|^2,
\end{align*}
which completes the proof.
\end{proof}

\setcounter{theorem}{0}
\renewcommand{\thetheorem}{B.\arabic{theorem}}
\setcounter{lemma}{0}
\renewcommand{\thelemma}{B.\arabic{lemma}}
\setcounter{proposition}{0}
\renewcommand{\theproposition}{B.\arabic{proposition}}
\setcounter{remark}{0}
\renewcommand{\theremark}{B.\arabic{remark}}

\section{Equivalence of algorithms}\label{app:algderiv}
We prove in Proposition~\ref{prop:kipld_standard_ud} that the \gls*{kipld} can be rewritten as a single underdamped Langevin diffusion evolving in $\bR^{d_\theta + N d_x}$. This brings the natural question about the algorithms we derived, namely, \gls*{kiplmc1} and \gls*{kiplmc2}, which are discretisations of the KIPLD, and how they relate to the standard discretisation schemes the rescaled diffusion in \eqref{eq:simplesde_scaled}. We show below that the standard discretisations of \eqref{eq:simplesde_scaled} are precisely the \gls*{kiplmc1} and \gls*{kiplmc2} algorithms we defined in Section~\ref{sec:algorithms}. This will later allow us to leverage existing results on the convergence of these standard discretisation schemes to obtain convergence results for our algorithms.

\glsunset{klmc1}

\subsection{KIPLMC1}\label{app:algderivKIPLMC1}
We show below the exponential integrator algorithm, presented in \citep{dalalyan2018sampling}, for our diffusion \eqref{eq:KIPLD} matches the \gls*{kiplmc1} algorithm we defined in Section~\ref{sec:algorithms}.

\begin{proposition}\label{prop:kipld_kiplmc1}
Let the sequence $(\widetilde{Z}_k,\widetilde{V}_k^z)_{k\geq 0}$ be the exponential-integrator discretisation of the rescaled underdamped diffusion \eqref{eq:simplesde_scaled} with friction parameter $\widetilde{\gamma}=\sqrt{N}\gamma$ and step-size $\widetilde{\eta}=\eta/\sqrt{N}$. Define the unscaled variables by
\begin{align*}
\widetilde{\theta}_k &= \theta_k,\qquad \widetilde{X}_k^i = N^{-1/2}X_k^i,\\
\widetilde{V}_k^\theta &= \sqrt{N}V_k^\theta,\qquad \widetilde{V}_k^{x_i} = V_k^{x_i},
\end{align*}
for each $i\in[N]$. Then, after the corresponding rescaling of the Gaussian noise variables, the resulting recursion in $(\theta_k,X_k^i,V_k^\theta,V_k^{x_i})_{i\in[N]}$ coincides exactly with the \gls*{kiplmc1} scheme defined in Section~\ref{subsec:kiplmc1}.
\end{proposition}
\begin{proof}
 Let us recall that, by Proposition~\ref{prop:kipld_standard_ud}, we can rewrite \eqref{eq:KIPLD} as the standard underdamped \gls*{sde}
\begin{align*}
\md \widetilde{\Z}_t &= \widetilde{\V}_t^z\,\md t, \\
\md \widetilde{\V}_t^z &= -\widetilde{\gamma}\,\widetilde{\V}_t^z\,\md t
- N\nabla_z \bar U_N(\widetilde{\Z}_t)\,\md t
+ \sqrt{2\widetilde{\gamma}}\,\md \widetilde{\B}_t,
\end{align*}
where $\tilde{\gamma} = \sqrt{N} \gamma$ is the friction parameter and we also implicitly adjust its initial conditioned to be rescaled. We discretise this diffusion using the exponential integrator scheme (termed \gls*{klmc1}) developed by \citep{dalalyan2018sampling}. In order to do this, define a sequence of functions $\tilde{\psi}_{k+1}^t = \int_0^t \tilde{\psi}_k^s \md s$ where $\tilde{\psi}_0^t = e^{-\tilde{\gamma} t}$. With this definition, the exponential integrator scheme with step-size $\tilde{\eta}$ is defined as
\begin{align*}
\widetilde{Z}_{k+1} &= \widetilde{Z}_k + \tilde{\psi}_1^{\tilde{\eta}} \tilde{V}_k - \tilde{\psi}_2^{\tilde{\eta}} N\nabla_z \bar{U}_N(\widetilde{Z}_k) + \sqrt{2\tilde{\gamma}}\zeta_{k+1}^\prime\\
\widetilde{V}_{k+1} &= \tilde{\psi}_0^{\tilde{\eta}} \widetilde{V}_k - \tilde{\psi}_1^{\tilde{\eta}} N \nabla_z \bar{U}_N(\tilde{Z}_k) + \sqrt{2\tilde{\gamma}} \zeta_{k+1},
\end{align*}
where the stacked vector $[(\zeta_k)_1, \dots, (\zeta_k)_{d_z},(\zeta'_k)_1, \dots (\zeta'_k)_{d_z}]^\top$ is a centred Gaussian random variable with covariance $\widetilde{C}\otimes I_{d_z}$, where $\widetilde{C}$ is given by
\begin{align*}
\widetilde{C} = \int_0^{\tilde{\eta}} \begin{bmatrix} \tilde{\psi}_0^t \\ \tilde{\psi}_1^t \end{bmatrix} \begin{bmatrix} \tilde{\psi}_0^t & \tilde{\psi}_1^t \end{bmatrix} \mathrm{d}t.
\end{align*}
Written explicitly via the definition of $\tilde{\psi}_i^t$, we have,
\small
\begin{equation*}
\widetilde{C}= \frac{1}{2\tilde{\gamma}} \begin{pmatrix}
    (1-e^{-2\tilde{\gamma}\tilde{\eta}}) &  \frac{(e^{-\tilde{\gamma}\tilde{\eta}} -1)^2}{\tilde{\gamma}} \\
    \frac{(e^{-\tilde{\gamma}\tilde{\eta}} -1)^2}{\tilde{\gamma}} &  \frac{2\tilde{\gamma}\tilde{\eta} + 4 e^{-\tilde{\gamma}\tilde{\eta}}- e^{-2\tilde{\gamma}\tilde{\eta}} -3}{\tilde{\gamma}^2}
\end{pmatrix}
\end{equation*}
\normalsize
Let us now relate the matrix $\widetilde{C}$ to the matrix we have defined in \eqref{eq:covariancemat} for the \gls*{kiplmc1} scheme. Written explicitly with the definitions in Section~\ref{subsec:kiplmc1}, \gls*{kiplmc1} instead defines functions $\psi_0^\eta = e^{-\gamma \eta}$ and $\psi_{i+1}^\eta = \int_0^\eta \psi_i^t \md t$, and the covariance matrix in \eqref{eq:covariancemat} is given as
\begin{align*}
C=\frac{1}{2\gamma}
\begin{pmatrix}
    (1-e^{- 2 \gamma\eta}) & \frac{(e^{- \gamma\eta}-1)^2}{\gamma}\\
    \frac{(e^{-\gamma\eta}-1)^2}{\gamma} & \frac{2\gamma\eta + 4e^{-\gamma\eta} - e^{-2\gamma\eta}-3}{\gamma^2}
\end{pmatrix}.
\end{align*}
Recall that the stacked vector $[(\varepsilon_k)_1, (\varepsilon'_k)_1, (\varepsilon_k)_2, (\varepsilon'_k)_2, \dots, ((\varepsilon_k)_{d_z}, (\varepsilon'_k)_{d_z})]^\top$ is a centred Gaussian random variable with covariance $C\otimes I_{d_z}$ as defined in Section~\ref{subsec:kiplmc1}. Note the relation between $\widetilde{C}$ and $C$ is given by
\begin{equation*}
    \widetilde{C}= \begin{pmatrix}
        N^{-\frac{1}{2}} C_{11} & N^{-1} C_{12}\\
        N^{-1} C_{21} & N^{-\frac{3}{2}} C_{22}
    \end{pmatrix}.
\end{equation*}
Hence, we recover
\begin{align*}
\zeta_{k} &= (N)^{-\frac{1}{4}} (\varepsilon_k^0, \dots, \varepsilon_k^N)^\top,\\
\zeta_k^\prime &= (N)^{-\frac{3}{4}} (\varepsilon^{0,\prime}_k, \dots, \varepsilon_k^{N,\prime})^\top.
\end{align*}
We first write the full exponential-integrator scheme coordinatewise, without yet substituting the relations between $\tilde{\eta}$ and $\eta$ or between the rescaled and original variables. This gives
\small
\begin{align*}
\widetilde{\theta}_{k+1} &= \widetilde{\theta}_k + \tilde{\psi}_1^{\tilde{\eta}} \widetilde{V}^\theta_k - \tilde{\psi}_2^{\tilde{\eta}} N \nabla_\theta \bar U_N(\widetilde{Z}_k) + \sqrt{2\tilde{\gamma}}(\zeta_{k}^\prime)^0,\\
\widetilde{X}_{k+1}^i &= \widetilde{X}_k^i + \tilde{\psi}_1^{\tilde{\eta}} \widetilde{V}_k^{x_i} - \tilde{\psi}_2^{\tilde{\eta}} N \nabla_{z_i} \bar U_N(\widetilde{Z}_k) + \sqrt{2\tilde{\gamma}}(\zeta_{k}^\prime)^i,\\
\widetilde{V}_{k+1}^\theta &= \tilde{\psi}_0^{\tilde{\eta}} \widetilde{V}_k^\theta - \tilde{\psi}_1^{\tilde{\eta}} N \nabla_\theta \bar U_N(\widetilde{Z}_k) + \sqrt{2\tilde{\gamma}}(\zeta_{k})^0,\\
\widetilde{V}_{k+1}^{x_i} &= \tilde{\psi}_0^{\tilde{\eta}} \widetilde{V}_k^{x_i} - \tilde{\psi}_1^{\tilde{\eta}} N \nabla_{z_i} \bar U_N(\widetilde{Z}_k) + \sqrt{2\tilde{\gamma}}(\zeta_{k})^i,
\end{align*}
\normalsize
where the noise variables $(\zeta_k)^0$, $(\zeta_k^\prime)^0$, $(\zeta_k)^i$ and $(\zeta_k^\prime)^i$ are the appropriate coordinates of $\zeta_k$ and $\zeta_k^\prime$ respectively. Now recall from Proposition~\ref{prop:kipld_standard_ud} that $\tilde{\gamma} = \sqrt{N}\gamma$, and choose $\tilde{\eta}=\eta/\sqrt{N}$. Then
\begin{align*}
\tilde{\psi}_0^{\tilde{\eta}} &= e^{-\gamma\eta} = \psi_0^\eta,\\
\tilde{\psi}_1^{\tilde{\eta}} &= \frac{1}{\gamma\sqrt{N}} (1-e^{-\gamma\eta}) = \frac{\psi_1^\eta}{\sqrt{N}},\\
\tilde{\psi}_2^{\tilde{\eta}} &= \frac{1}{\gamma N}(\eta - \frac{1}{\gamma}(1-e^{-\gamma\eta})) = \frac{\psi_2^\eta}{N}.
\end{align*}
Using also the identities
\begin{equation}
\begin{aligned}
N\nabla_\theta \bar U_N(\widetilde{Z}_k) &= \sum_{i=1}^N \nabla_\theta U(\widetilde{\theta}_k,\sqrt{N}\widetilde{X}_k^i),\\
N\nabla_{z_i} \bar U_N(\widetilde{Z}_k) &= \sqrt{N}\nabla_x U(\widetilde{\theta}_k,\sqrt{N}\widetilde{X}_k^i),
\end{aligned}\label{eq:grad_rescaling}
\end{equation}
together with $\widetilde{\theta}_k=\theta_k$, $\widetilde{X}_k^i=N^{-1/2}X_k^i$, $\widetilde{V}_k^\theta=\sqrt{N}V_k^\theta$, $\widetilde{V}_k^{x_i}=V_k^{x_i}$, $(\zeta_k)^0=N^{-1/4}\varepsilon_k^0$, $(\zeta_k^\prime)^0=N^{-3/4}\varepsilon_k^{0,\prime}$, $(\zeta_k)^i=N^{-1/4}\varepsilon_k^i$, and $(\zeta_k^\prime)^i=N^{-3/4}\varepsilon_k^{i,\prime}$, we obtain for the $\theta$-coordinates
\small
\begin{align*}
\theta_{k+1} &= \theta_k + \frac{\psi_1^\eta}{\sqrt{N}}\widetilde{V}^\theta_k - \frac{\psi_2^\eta}{N} \sum_{i=1}^N \nabla_\theta U(\theta_k, X_k^i)+ \sqrt{\frac{2\gamma}{N}} \varepsilon^{0,\prime}_{k},\\
\widetilde{V}^\theta_{k+1} &= \psi_0^\eta \widetilde{V}^\theta_k - \frac{\psi_1^\eta}{\sqrt{N}} \sum_{i=1}^N \nabla_\theta U(\theta_k, X_k^i) + \sqrt{2\gamma} \varepsilon^0_{k}.
\end{align*}
\normalsize
Since $\widetilde{V}_k^\theta=\sqrt{N}V_k^\theta$, substituting this relation into the first line and dividing the second line by $\sqrt{N}$ yields
\small
\begin{align*}
\theta_{k+1} &= \theta_k + \psi_1^\eta V^\theta_k - \frac{\psi_2^\eta}{N} \sum_{i=1}^N \nabla_\theta U(\theta_k, X_k^i)+ \sqrt{\frac{2\gamma}{N}} \varepsilon^{0,\prime}_{k},\\
V^\theta_{k+1} &= \psi_0^\eta V^\theta_k - \frac{\psi_1^\eta}{N} \sum_{i=1}^N \nabla_\theta U(\theta_k, X_k^i) + \sqrt{\frac{2\gamma}{N}} \varepsilon^0_{k}.
\end{align*}
\normalsize
For each particle coordinate $i\in[N]$, the same substitutions give
\small
\begin{align*}
\widetilde{X}_{k+1}^i &= \widetilde{X}_k^i + \frac{\psi_1^\eta}{\sqrt{N}}V^{x_i}_k - \frac{\psi_2^\eta}{\sqrt{N}} \nabla_x U(\theta_k, X_k^i)+ \sqrt{\frac{2\gamma}{N}} \varepsilon^{i,\prime}_{k},\\
V^{x_i}_{k+1} &= \psi_0^\eta V^{x_i}_k - \psi_1^\eta \nabla_x U(\theta_k, X_k^i) + \sqrt{2\gamma} \varepsilon^i_{k}.
\end{align*}
\normalsize
Finally, multiplying the first equation by $\sqrt{N}$, i.e. using $X_k^i=\sqrt{N}\widetilde{X}_k^i$, yields
\small
\begin{align*}
X^i_{k+1} &= X^i_k + \psi_1^\eta V^{x_i}_k - \psi_2^\eta \nabla_x U(\theta_k, X_k^i)+ \sqrt{2\gamma} \varepsilon^{i,\prime}_{k}\\
V^{x_i}_{k+1} &= \psi_0^\eta V^{x_i}_k - \psi_1^\eta \nabla_x U(\theta_k, X_k^i) + \sqrt{2\gamma} \varepsilon^i_{k}.
\end{align*}
\normalsize
Together with the $\theta$-update above, this is precisely the \gls*{kiplmc1} scheme.
\end{proof}

\subsection{KIPLMC2}\label{app:algderivKIPLMC2}
\begin{proposition}\label{prop:kipld_kiplmc2}
Let $(\widetilde{Z}_k,\widetilde{V}_k^z)_{k\geq 0}$ be the OBABO discretisation of the rescaled underdamped diffusion \eqref{eq:simplesde_scaled} with friction parameter $\widetilde{\gamma}=\sqrt{N}\gamma$ and step-size $\widetilde{\eta}=\eta/\sqrt{N}$. Define the unscaled variables by
\begin{align*}
\widetilde{\theta}_k &= \theta_k,\qquad \widetilde{X}_k^i = N^{-1/2}X_k^i,\\
\widetilde{V}_k^\theta &= \sqrt{N}V_k^\theta,\qquad \widetilde{V}_k^{x_i} = V_k^{x_i},
\end{align*}
for each $i\in[N]$. Then the resulting recursion in $(\theta_k,X_k^i,V_k^\theta,V_k^{x_i})_{i\in[N]}$ coincides exactly with the \gls*{kiplmc2} scheme defined in Algorithm~\ref{alg:KIPLMC2}.
\end{proposition}
\begin{proof}
In this case we again apply Proposition~\ref{prop:kipld_standard_ud} to apply the OBABO results from \citep{monmarche2021high} for the standard underdamped diffusion, to our own. Let us recall the standard form, given in Proposition~\ref{prop:kipld_standard_ud} with $\tilde{\gamma}=\sqrt{N}\gamma$,
\begin{align*}
\md \widetilde{\Z}_t &= \widetilde{\V}_t^z\,\md t, \\
\md \widetilde{\V}_t^z &= -\widetilde{\gamma}\,\widetilde{\V}_t^z\,\md t
- N\nabla_z \bar U_N(\widetilde{\Z}_t)\,\md t
+ \sqrt{2\widetilde{\gamma}}\,\md \widetilde{\B}_t,
\end{align*}
This diffusion is discretised with the OBABO scheme with step-size $\tilde{\eta}$, to obtain,
\small
\begin{align*}
\widetilde{Z}_{k+1} &= \widetilde{Z}_k + \tilde{\eta} \left(\tilde{\delta}\widetilde{V}_k + \sqrt{1-\tilde{\delta}^2}\xi_{k+1}\right)-\frac{\tilde{\eta}^2}{2} N\nabla_z\bar{U}_N(\widetilde{Z}_k)\\
\widetilde{V}_{k+1}&= \tilde{\delta}^2 \widetilde{V}_{k} + \sqrt{1-\tilde{\delta}^2} (\tilde{\delta}\xi_{k+1} + \xi^\prime_{k+1})-\frac{\tilde{\eta}\tilde{\delta}}{2} N (\nabla_z \bar{U}_N(\widetilde{Z}_k) + \nabla_z \bar{U}_N(\widetilde{Z}_{k+1})),
\end{align*}
\normalsize
where we recall that $\tilde{\delta}=e^{-\tilde{\gamma}\tilde{\eta}/2}$. Let us begin at looking at the discretisation coordinatewise in the scaled case,
\small
\begin{align*}
\widetilde{\theta}_{k+1} &= \widetilde{\theta}_k + \tilde{\eta}\left(\tilde{\delta} \widetilde{V}^\theta_{k} + \sqrt{1-\tilde{\delta}^2} \xi^0_{k+1}\right) - \frac{\tilde{\eta}^2}{2} N \nabla_\theta \bar{U}_N(\widetilde{Z}_k),\\
\widetilde{X}^i_{k+1} &= \widetilde{X}^i_k + \tilde{\eta}\left(\tilde{\delta} \widetilde{V}^{x_i}_{k} + \sqrt{1-\tilde{\delta}^2} \xi^i_{k+1}\right) - \frac{\tilde{\eta}^2}{2} N \nabla_{z_i} \bar{U}_N(\widetilde{Z}_k),\\
\widetilde{V}^\theta_{k+1} &= \tilde{\delta}^2 \widetilde{V}^\theta_k + \sqrt{1-\tilde{\delta}^2} (\tilde{\delta} \xi^0_{k+1} + \xi_{k+1}^{0,\prime}) - \frac{\tilde{\eta}\tilde{\delta}}{2} N(\nabla_\theta \bar{U}_N(\widetilde{Z}_k) + \nabla_\theta \bar{U}_N(\widetilde{Z}_{k+1}),\\
\widetilde{V}^{x_i}_{k+1} &= \tilde{\delta}^2 \widetilde{V}^{x_i}_k + \sqrt{1-\tilde{\delta}^2} (\tilde{\delta} \xi^i_{k+1} + \xi_{k+1}^{i,\prime}) - \frac{\tilde{\eta}\tilde{\delta}}{2} N(\nabla_{z_i} \bar{U}_N(\widetilde{Z}_k) + \nabla_{z_i} \bar{U}_N(\widetilde{Z}_{k+1})).
\end{align*}
\normalsize
We now proceed to rescaling the algorithm by observing that $\tilde{\delta}=\delta$ and recalling that $\tilde{\psi}_0^{\tilde{\eta}} = \psi_0^\eta$. Using this, $\tilde{\eta} = \eta/ \sqrt{N}$ and $\tilde{\gamma} = \sqrt{N}\gamma$ and the identities \eqref{eq:grad_rescaling}, we obtain for the $\theta$-coordinates,
\small
\begin{align*}
\theta_{k+1} &= \theta_k + \frac{\eta}{\sqrt{N}} \left(\delta \widetilde{V}^\theta_k + \sqrt{1-\delta^2}\xi_{k+1}\right) - \frac{\eta^2}{2N}\sum_{i=1}^N\nabla_\theta U(\theta_k, X_k^i)\\
\widetilde{V}^\theta_{k+1} &= \delta^2 \widetilde{V}^\theta_k  + \sqrt{1-\delta^2} (\delta\xi_{k+1} + \xi^\prime_{k+1})
- \frac{\eta\delta}{2\sqrt{N}}\sum_{i=1}^N \bigg(\nabla_\theta U(\theta_k, X_k^i) + \nabla_\theta U(\theta_{k+1}, X_{k+1}^i)\bigg).
\end{align*}
\normalsize
Substituting in the relationship $\widetilde{V}_k^\theta=\sqrt{N}V_k^\theta$ and $\widetilde{X}_k^i = X_k^i/\sqrt{N}$, we obtain,
\small
\begin{align*}
\theta_{k+1} &= \theta_k + \eta \left(\delta V^\theta_k + \sqrt{\frac{1-\delta^2}{N}}\xi_{k+1}\right)
- \frac{\eta^2}{2N}\sum_{i=1}^N\nabla_\theta U(\theta_k, X_k^i)\\
V^\theta_{k+1} &= \delta^2 V^\theta_k  + \sqrt{\frac{1-\delta^2}{N}} (\delta\xi_{k+1} + \xi^\prime_{k+1}) 
- \frac{\eta\delta}{2N}\sum_{i=1}^N \bigg(\nabla_\theta U(\theta_k, X_k^i) + \nabla_\theta U(\theta_{k+1}, X_{k+1}^i)\bigg).
\end{align*}
\normalsize
For the particles $i\in [N]$, the same substitutions give,
\small
\begin{align*}
\widetilde{X}^i_{k+1} &= \widetilde{X}^i_k + \frac{\eta}{\sqrt{N}}\left(\delta V^{x_i}_{k} + \sqrt{1-\delta^2} \xi^i_{k+1}\right) - \frac{\eta^2}{2\sqrt{N}} \nabla_{x_i} U(\theta_k , \sqrt{N}\widetilde{X}_k^i),\\
V^{x_i}_{k+1} &= \delta^2 V^{x_i}_k + \sqrt{1-\delta^2} (\delta \xi^i_{k+1} + \xi_{k+1}^{i,\prime}) - \frac{\eta\delta}{2} (\nabla_{x_i} U(\theta_k, \sqrt{N}\widetilde{X}^i_k)+ \nabla_{x_i} U(\theta_{k+1}, \sqrt{N}\widetilde{X}^i_{k+1})).
\end{align*}
\normalsize
Finally, using $X_k^i = \sqrt{N}\widetilde{X}_k^i$,
we recover,
\small
\begin{align*}
X^i_{k+1} &= X^i_k + \eta \left(\delta V^{x_i}_k + \sqrt{1-\delta^2}\xi^i_{k+1}\right)- \frac{\eta^2}{2}\nabla_x U(\theta_k, X_k^i)\\
V^{x_i}_{k+1} &= \delta^2 V^{x_i}_k  + \sqrt{1-\delta^2} (\delta\xi^i_{k+1} + \xi^{i,\prime}_{k+1}) - \frac{\eta\delta}{2} \bigg(\nabla_x U(\theta_k, X_k^i) + \nabla_x U(\theta_{k+1}, X_{k+1}^i)\bigg).
\end{align*}
\normalsize
This shows that the OBABO discretisation of the rescaled diffusion, with step-size $\tilde{\eta} = \eta/\sqrt{N}$, is precisely the \gls*{kiplmc2} scheme.
\end{proof}
\begin{remark} The results in Propositions~\ref{prop:kipld_standard_ud}--\ref{prop:kipld_kiplmc2} significantly streamline our proofs, since they allow us to connect \gls*{kipld} and our methods \gls*{kiplmc1} and \gls*{kiplmc2} to the standard underdamped diffusion and its discretisations, for which there is a rich literature of results that we can apply.
\end{remark}
\section{Proofs of main results}\label{app:proofs}

\subsection{Proof of Proposition~\ref{prop:stationary}}\label{app:proofstationary}
Given Proposition~\ref{prop:kipld_standard_ud}, the system \eqref{eq:simplesde_scaled} has a stationary measure given by the density
\begin{equation*}
\tilde{\pi}(z, \bar{v}) \propto \exp \left(-N\bar{U}_N(z) - \frac{1}{2}\|\bar{v}\|^2 \right),
\end{equation*}
where $\bar{U}_N$ is given in \eqref{eq:barU}. Let us now look at $\theta$-marginal of this density, which can be written as
\begin{equation*}
\pi_\Theta (\theta) \propto \int
e^{-\sum_{i=1}^N U(\theta, \sqrt{N} z_i) - \frac{1}{2} \|\bar{v}^z\|^2}  \md z_x \md \bar{v}.
\end{equation*}
where $z_x = (z_1, \ldots, z_N) \in \bR^{N d_x}$. Note now that, integrating out $\bar{v}$, using a change of variables $x_i' = \sqrt{N} z_i$, and setting $x' = (x_1', \ldots, x_N')$, we have
\begin{align*}
\pi_\Theta (\theta) &\propto\int_{\mathbb{R}^{Nd_x}} 
e^{-\sum_{i=1}^N U(\theta, x'_i)} \md x', \\
&\propto \left(\int_{\bR^{d_x}} e^{-U(\theta, x') } \md x'\right)^N\\
&= \exp(N \log p_\theta(y)),
\end{align*}
since $p_\theta(y) = \int e^{-U(\theta, x)} \md x$ by definition, where $U(\theta, x) = -\log p_\theta(x, y)$.

\subsection{Proof of Proposition~\ref{prop:mmleerror}}\label{app:proofmmleerror}
Note that, by \ref{assump:strongconvexity}, $U(\theta, x)$ is jointly $\mu$-strongly convex. Let $\kappa(\theta) = - \log p_\theta(y)$. Using the Pr\'{e}kopa-Leindler inequality for strongly log-concave distributions \citep[Theorem~3.8]{saumard2014log}, we can see that
\begin{align*}
\langle \theta - \theta', \nabla \kappa(\theta) - \nabla \kappa(\theta') \rangle \geq \mu \|\theta - \theta'\|^2.
\end{align*}
The bound now follows using Lemma~A.8 from \citep{altschuler2023fasterhighaccuracylogconcavesampling}.

\subsection{Proof of Proposition~\ref{prop:kipld_conv}}\label{app:proofkipld_conv}
By Proposition~\ref{prop:kipld_standard_ud}, the rescaled process
\((\widetilde{\Z}_t,\widetilde{\V}_t^z)\) solves the standard underdamped \gls*{sde}
\begin{equation}\label{eq:proof_kipld_scaled}
\begin{aligned}
\md \widetilde{\Z}_t &= \widetilde{\V}_t^{z}\md t,\\
\md \widetilde{\V}_t^z
&= -\widetilde{\gamma}\widetilde{\V}_t^z\md t
- N \nabla \bar{U}_N(\widetilde{\Z}_t)\md t
+ \sqrt{2\widetilde{\gamma}}\,\md \widetilde{\B}_t,
\end{aligned}
\end{equation}
with $\widetilde{\gamma}=\sqrt{N}\gamma$. By Lemmas~\ref{lem:strongconvexity_barU} and \ref{lem:gradientlipschitz_barU},
\(\bar U_N\) is \(\mu\)-strongly convex and \(L\)-gradient-Lipschitz, hence
\(N\bar{U}_N\) is \(N\mu\)-strongly convex and \(NL\)-gradient-Lipschitz. Define
\[
\widetilde{\mu}:=N\mu,\qquad \widetilde{L}:=NL.
\]
The condition \(\gamma\ge \sqrt{\mu+L}\) implies
\[
\widetilde{\gamma}=\sqrt{N}\gamma \ge \sqrt{N(\mu+L)}=\sqrt{\widetilde{\mu}+\widetilde{L}}.
\]
Therefore Theorem~1 in \citep{dalalyan2018sampling} applies to
\eqref{eq:proof_kipld_scaled} and yields, for \(s\ge 0\),
\begin{align*}
W_2\!\left(\mathcal{L}(\widetilde{\Z}_s),\widetilde\pi_Z\right)
&\le\;
\sqrt{2}\exp\!\left(-\frac{\widetilde{\mu}}{\widetilde{\gamma}}s\right)
W_2\!\left(\mathcal{L}(\widetilde{\Z}_0),\widetilde\pi_Z\right).
\end{align*}
where \(\widetilde\pi_Z\) is the \(\widetilde{\Z}\)-marginal of \(\widetilde{\pi}\)
(Lemma~\ref{lem:scaledstat}).

Now relate rescaled and original time: by definition,
\[
\widetilde{\Z}_{t/\sqrt{N}}
=\big(\bm{\theta}_{t},N^{-1/2}\X_t^1,\ldots,N^{-1/2}\X_t^N\big)^\intercal.
\]
Hence the \(\theta\)-component of \(\widetilde{\Z}_{t/\sqrt{N}}\) is \(\bm{\theta}_t\), and
the \(\theta\)-marginal of \(\widetilde\pi_Z\) is \(\pi_\Theta\). We note
\begin{equation*}
W_2\!\left(\mathcal{L}(\bm{\theta}_t),\pi_\Theta\right)
\le
W_2\!\left(\mathcal{L}(\widetilde{\Z}_{t/\sqrt{N}}),\widetilde\pi_Z\right).
\end{equation*}
Using \(s=t/\sqrt{N}\) gives
\[
\frac{\widetilde{\mu}}{\widetilde{\gamma}}\,\frac{t}{\sqrt{N}}
=\frac{N\mu}{\sqrt{N}\gamma}\,\frac{t}{\sqrt{N}}
=\frac{\mu}{\gamma}t.
\]
Therefore
\begin{equation*}
W_2\!\left(\mathcal{L}(\bm{\theta}_t),\pi_\Theta\right)
\le
\sqrt{2}\exp\!\left(-\frac{\mu}{\gamma}t\right)
W_2\!\left(\mathcal{L}(\widetilde{\Z}_0),\widetilde\pi_Z\right).
\end{equation*}
Finally, with \(\bar Z_\star\sim \widetilde\pi_Z\),
\[
W_2\!\left(\mathcal{L}(\widetilde{\Z}_0),\widetilde\pi_Z\right)
\le \mathbb{E}\!\left[\|\widetilde{\Z}_0-\bar Z_\star\|^2\right]^{1/2},
\]
which is exactly the stated bound (up to the same notation
\(\widetilde{\Z}_0\equiv \Z_0\) used in the proposition).

\subsection{Proof of Theorem~\ref{thm:KPLMC1}}\label{app:proofnumericalKPLMC1}
By Proposition~\ref{prop:kipld_standard_ud}, we can rewrite KIPLD as the standard underdamped diffusion with stationary measure
\begin{align*}
\widetilde{\pi}(z, v) \propto e^{-N \bar{U}_N(z) - \|v\|^2/2}.
\end{align*}
We will now look at the standard kinetic Langevin Monte Carlo discretisation of the \gls*{sde} and show that this scheme coincides with \gls*{kiplmc1} to utilize the bounds in \citep{dalalyan2018sampling} directly.

Let us write the Exponential Integrator discretisation of this scheme as we do in \ref{app:algderivKIPLMC1},
\small
\begin{equation}\label{eq:underdamped_normal_1}
\begin{aligned}
\widetilde{Z}_{n+1} &= \widetilde{Z}_{n} + \widetilde{\psi}_1^{\tilde{\eta}} \widetilde{V}_n^z - \widetilde{\psi}_2^{\tilde{\eta}} N \nabla \bar{U}_N(\widetilde{Z}_n)+ \sqrt{2 \tilde{\gamma}} \zeta_{n+1}',\\
\widetilde{V}_{n+1}^z &= \widetilde{\psi}_0^{\tilde{\eta}} \widetilde{V}_n^z - \widetilde{\psi}_1^{\tilde{\eta}} N \nabla \bar{U}_N(\widetilde{Z}_n) + \sqrt{2 \tilde{\gamma}} \zeta_{n+1}. \end{aligned}
\end{equation}
\normalsize
Since the function $z \mapsto N \bar{U}_N(z)$ is $N\mu$ strongly convex and $NL$-gradient-Lipschitz, for $\tilde{\gamma} \geq \sqrt{N\mu + N L}$ and $\tilde{\eta} \leq \mu / (4 \tilde{\gamma} L)$, \citep[Theorem~2]{dalalyan2018sampling} directly implies that
\small
\begin{align*}
W_2(\nu_n, \widetilde\pi_Z) &\leq \sqrt{2} \mathbb{E}[\|Z_0-\bar{Z}_\star\|^2]^{1/2} \left( 1 - \frac{0.75 \tilde{\mu} \tilde{\eta}}{\tilde{\gamma}}\right)^n + \frac{\tilde{\eta}\tilde{L}}{\tilde{\mu}}\sqrt{2d_z},
\end{align*}
\normalsize
where $\nu_n:=\mathcal L(\widetilde Z_n)$ denotes the law of the position iterate in \eqref{eq:underdamped_normal_1}, $\widetilde\pi_Z$ denotes the position marginal of the stationary measure in Lemma~\ref{lem:scaledstat}, and $\bar Z_\star\sim\widetilde\pi_Z$. We now recall from \ref{app:algderivKIPLMC1}, that $\tilde{\eta}=\eta/\sqrt{N}$ and $\tilde{\gamma}=\sqrt{N}\gamma$, from which we obtain,
\small
\begin{equation*}
W_2(\nu_n, \widetilde\pi_Z) \leq C_1\left(1-\frac{0.75 \mu\eta}{\gamma}\right)^n + \eta C_2,
\end{equation*}
\normalsize
where,
\small
\begin{align*}
C_1&= \sqrt{2} \bE[\|Z_0-\bar{Z}_\star\|^2]^{1/2}, \qquad \bar{Z}_\star\sim \widetilde\pi_Z\\
C_2&=\sqrt{2}\frac{L}{\mu} \sqrt{\frac{d_z}{N}}.
\end{align*}
\normalsize

The rescaling leaves the $\theta$-coordinate unchanged, so the $\theta$-component of $\widetilde Z_n$ is $\theta_n$. Consequently, $\mathcal L(\theta_n)$ and $\pi_\Theta$ are the $\theta$-marginals of $\nu_n$ and $\widetilde\pi_Z$, respectively. Since the coordinate projection $z\mapsto\theta$ is $1$-Lipschitz, $W_2(\mathcal L(\theta_n),\pi_\Theta)\leq W_2(\nu_n,\widetilde\pi_Z)$. Applying this inequality and the triangle inequality, we obtain
\small
\begin{align*}
\mathbb{E}[\|\theta_n-\bar{\theta}_\star\|^2]^{1/2} &= W_2(\mathcal{L}(\theta_n), \delta_{\bar{\theta}_\star})\\
&\leq W_2(\mathcal{L}(\theta_n), \pi_\Theta) + W_2(\pi_\Theta, \delta_{\bar{\theta}_\star})\\
&\leq C_1 \left(1-\frac{0.75 \mu\eta}{\gamma}\right)^n + \eta C_2 + \sqrt{\frac{C_3}{N}},
\end{align*}
\normalsize
where the last term is bounded by Proposition~\ref{prop:mmleerror} with $C_3=d_\theta/\mu$.

\subsection{Proof of Theorem~\ref{thm:KIPLMC2}}\label{app:proofnumericalKPLMC2}
Again, we consider our rescaled standard underdamped diffusion, given by \eqref{eq:simplesde_scaled} and seek to directly apply the results from \citep{monmarche2021high} for our scheme.

We now write the splitting scheme as in \ref{app:algderivKIPLMC2},
\small
\begin{equation}
\begin{aligned}\label{eq:underdamped_normal_2}
\widetilde{Z}_{n+1} &= \widetilde{Z}_n + \tilde{\eta} (\tilde{\delta} \widetilde{V}^z_n + \sqrt{1-\tilde{\delta}^2} \xi_{n+1})- \frac{\tilde{\eta}^2}{2}N \nabla \bar{U}_N(\widetilde{Z}_n)\\
\widetilde{V}^z_{n+1} &= \tilde{\delta}^2 \widetilde{V}_n^z - \frac{\tilde{\eta}\tilde{\delta}}{2}(N\nabla \bar{U}_N(\widetilde{Z}_n) + N \nabla\bar{U}_N (\widetilde{Z}_{n+1}) ) + \sqrt{1-\tilde{\delta}^2}(\tilde{\delta}\xi_{n+1} + \xi_{n+1}^\prime),
\end{aligned}
\end{equation}
\normalsize
where $\tilde{\delta}=e^{-\tilde{\gamma}\tilde{\eta}/2}$.  Set
\begin{align*}
\widetilde\mu&:=N\mu,
&\widetilde L&:=NL,\\
\widetilde\kappa&:=\frac{\widetilde\mu}{3\widetilde\gamma},
&A_N&:=\sqrt{3}\bigl(\sqrt{NL}\lor(NL)^{-1/2}\bigr).
\end{align*}
By Lemmas~\ref{lem:strongconvexity_barU} and~\ref{lem:gradientlipschitz_barU}, the potential $N\bar U_N$ is $\widetilde\mu$-strongly convex and has $\widetilde L$-Lipschitz gradient.  Moreover,
\begin{align*}
\widetilde\gamma\geq2\sqrt{\widetilde L}
&\quad\Longleftrightarrow\quad \gamma\geq2\sqrt L,\\
\widetilde\eta\leq\frac{\widetilde\mu}{33\widetilde\gamma^3}
&\quad\Longleftrightarrow\quad \eta\leq\frac{\mu}{33\gamma^3},\\
1-\widetilde\eta\widetilde\kappa
&=1-\frac{\eta\mu}{3\gamma}=:q.
\end{align*}
Indeed, $\mu\leq L$ and the stated parameter restrictions give
$0<\eta\mu/(3\gamma)\leq\mu^2/(99\gamma^4)\leq1/1584$, so $q\in(0,1)$.
Let $\widetilde P$ be the transition kernel of \eqref{eq:underdamped_normal_2}, let $\widetilde\pi_{\widetilde\eta}$ be its invariant probability measure on phase space, let $\widetilde\pi$ be the invariant probability measure of the diffusion in Lemma~\ref{lem:scaledstat}, and let $\widetilde\pi_Z$ be the position marginal of $\widetilde\pi$.  Theorem~1 of \citep{monmarche2021high} gives, for all $\rho,\sigma\in\mathcal P_2(\mathbb R^{2d_z})$,
\begin{align}
W_2(\rho\widetilde P^n,\sigma\widetilde P^n)
\leq A_Nq^{n/2}W_2(\rho,\sigma).
\label{eq:kiplmc2_phase_contraction}
\end{align}
Proposition~11 of the same reference, applied with $p=2$, gives
\begin{align}
b_N
&:=W_2(\widetilde\pi_{\widetilde\eta},\widetilde\pi)\nonumber\\
&\leq\widetilde\eta\sqrt{d_z}\,
\frac{2A_N\widetilde K}{\widetilde\kappa}
=\eta\sqrt{d_z}\,\frac{6\gamma\widetilde K}{N\mu}\,A_N
=:B_N,
\label{eq:kiplmc2_invariant_bias}
\end{align}
where
\begin{align*}
\widetilde K
:=\widetilde L\left(1+e^{\widetilde L\widetilde\eta^2}
\left(\frac{\widetilde\eta}{6}
+\frac{\widetilde\eta^2\widetilde L}{24}\right)\right)
\left(1+\frac{\widetilde\eta\widetilde L}{2\sqrt{\widetilde\mu}}\right).
\end{align*}

Substituting $\widetilde L=NL$, $\widetilde\eta=\eta/\sqrt N$, and $\widetilde\mu=N\mu$ into $\widetilde K$ yields the exact identity
\begin{align*}
\widetilde K
&=NL\left(1+e^{L\eta^2}
\left(\frac{\eta}{6\sqrt N}+\frac{\eta^2L}{24}\right)\right)\left(1+\frac{\eta L}{2\sqrt\mu}\right).
\end{align*}
Define
\begin{align*}
K:=L\left(1+e^{L\eta^2}
\left(\frac{\eta}{6}+\frac{\eta^2L}{24}\right)\right)
\left(1+\frac{\eta L}{2\sqrt\mu}\right).
\end{align*}
Since $N\geq1$,
\begin{align*}
\widetilde K&\leq NK,
&A_N\leq\sqrt{3N}\bigl(\sqrt L\lor L^{-1/2}\bigr).
\end{align*}
For conciseness, set $M_L:=\sqrt L\lor L^{-1/2}$ and define
\begin{align*}
G_\eta
&:=\frac{6\gamma K}{\mu}M_L\\
&=\frac{6\gamma L}{\mu}M_L \left(1+\frac{\eta L}{2\sqrt\mu}\right)\left(1+e^{L\eta^2}
\left(\frac{\eta}{6}+\frac{\eta^2L}{24}\right)\right).
\end{align*}
To make the implicit constant in Theorem~\ref{thm:KIPLMC2} uniform in $\eta$, set
\begin{align*}
\eta_0&:=\frac{\mu}{33\gamma^3},\\
G_\star
&:=\frac{6\gamma L}{\mu}M_L\left(1+\frac{\eta_0L}{2\sqrt\mu}\right)\left(1+e^{L\eta_0^2}
\left(\frac{\eta_0}{6}+\frac{\eta_0^2L}{24}\right)\right).
\end{align*}
Every factor in $G_\eta$ is nondecreasing for $\eta\geq0$.  Hence the restriction $0<\eta\leq\eta_0$ gives $G_\eta\leq G_\star$, where $G_\star$ depends only on $\mu$, $L$, and $\gamma$.  Consequently, \eqref{eq:kiplmc2_invariant_bias} gives
\begin{align}
B_N
&\leq\eta\sqrt{d_z}\,\frac{6\gamma K}{\mu}\,
\sqrt{3N}\,M_L
\leq\sqrt3\,G_\star\eta\sqrt{Nd_z}.
\label{eq:kiplmc2_coarse_bias}
\end{align}

Let $\varphi:=\mathcal N(0,\mathbb I_{d_z})$ and let $\rho_n$ be the law of the full rescaled phase-space iterate $(\widetilde Z_n,\widetilde V_n^z)$.  The initialisation in Theorem~\ref{thm:KIPLMC2} and the rescaling in Proposition~\ref{prop:kipld_kiplmc2} imply that
\begin{align*}
\rho_0=\nu\otimes\varphi.
\end{align*}
Lemma~\ref{lem:scaledstat} gives $\widetilde\pi=\widetilde\pi_Z\otimes\varphi$.  Coupling the two identical velocity marginals, and using the position projection for the reverse inequality, gives
\begin{align*}
W_2(\rho_0,\widetilde\pi)
=W_2(\nu,\widetilde\pi_Z)=:D_0.
\end{align*}
Since $\widetilde\pi_{\widetilde\eta}$ is invariant for $\widetilde P$, \eqref{eq:kiplmc2_phase_contraction}, \eqref{eq:kiplmc2_invariant_bias}, and two applications of the triangle inequality yield
\begin{align*}
W_2(\rho_n,\widetilde\pi)
&\leq W_2(\rho_0\widetilde P^n,
\widetilde\pi_{\widetilde\eta}\widetilde P^n)+b_N\\
&\leq A_Nq^{n/2}W_2(\rho_0,\widetilde\pi_{\widetilde\eta})+b_N\\
&\leq A_Nq^{n/2}(D_0+b_N)+b_N\\
&\leq A_Nq^{n/2}(D_0+B_N)+B_N\\
&\leq\sqrt3\,M_L\sqrt N\,D_0q^{n/2}+B_N+\sqrt3\,M_L\sqrt N\,q^{n/2}B_N\\
&\leq\sqrt3\,M_L\sqrt N\,D_0q^{n/2}
+\sqrt3\,G_\star\eta\sqrt{Nd_z}+3M_LG_\star\eta N\sqrt{d_z}\,q^{n/2}.
\end{align*}
Thus, with
\begin{align*}
C_\star:=\max\{\sqrt3M_L,\sqrt3G_\star,3M_LG_\star\},
\end{align*}
we have the fully explicit intermediate bound
\begin{align}
W_2(\rho_n,\widetilde\pi)
&\leq C_\star\left(\sqrt N\,D_0q^{n/2}+\eta N\sqrt{d_z}\,q^{n/2}+\eta\sqrt{Nd_z}\right).
\label{eq:kiplmc2_main_dependencies}
\end{align}

The map $(z,v)\mapsto\theta$ is $1$-Lipschitz, the rescaling leaves the $\theta$-coordinate unchanged, and its pushforwards of $\rho_n$ and $\widetilde\pi$ are $\mathcal L(\theta_n)$ and $\pi_\Theta$, respectively.  Therefore,
\begin{align*}
W_2(\mathcal L(\theta_n),\pi_\Theta)
\leq W_2(\rho_n,\widetilde\pi).
\end{align*}
Finally, the triangle inequality and Proposition~\ref{prop:mmleerror} give
\begin{align*}
\bE[\|\theta_n-\bar\theta_\star\|^2]^{1/2}
&=W_2(\mathcal L(\theta_n),\delta_{\bar\theta_\star})\\
&\leq W_2(\rho_n,\widetilde\pi)
+\frac{1}{\sqrt\mu}\sqrt{\frac{d_\theta}{N}}.
\end{align*}
Moreover, since $q=1-\eta\mu/(3\gamma)\in(0,1)$,
\begin{align*}
q^{n/2}
&\leq\exp\left(-\frac{\mu\eta n}{6\gamma}\right).
\end{align*}
Combining this inequality with \eqref{eq:kiplmc2_main_dependencies}, substituting $d_z=d_\theta+Nd_x$, and absorbing $C_\star$ and $\mu^{-1/2}$ into an implicit constant gives, with $\mathcal E_n:=\bE[\|\theta_n-\bar\theta_\star\|^2]^{1/2}$,
\begin{align*}
\mathcal E_n
&\lesssim \left(\sqrt N\,W_2(\nu,\widetilde\pi_Z)
+ \eta N\sqrt{d_\theta+Nd_x}\right)\exp\left(-\frac{\mu\eta n}{6\gamma}\right)+\eta\sqrt{N(d_\theta+Nd_x)}
+\sqrt{\frac{d_\theta}{N}}.
\end{align*}
The implicit constant depends only on $\mu$, $L$, and $\gamma$, and this is the claimed bound.

\section{Results}

\subsection{Multi-well Experiment}\label{sec:multiwell}

\begin{figure*}[t!]
    \centering
    \includegraphics[width=\linewidth]{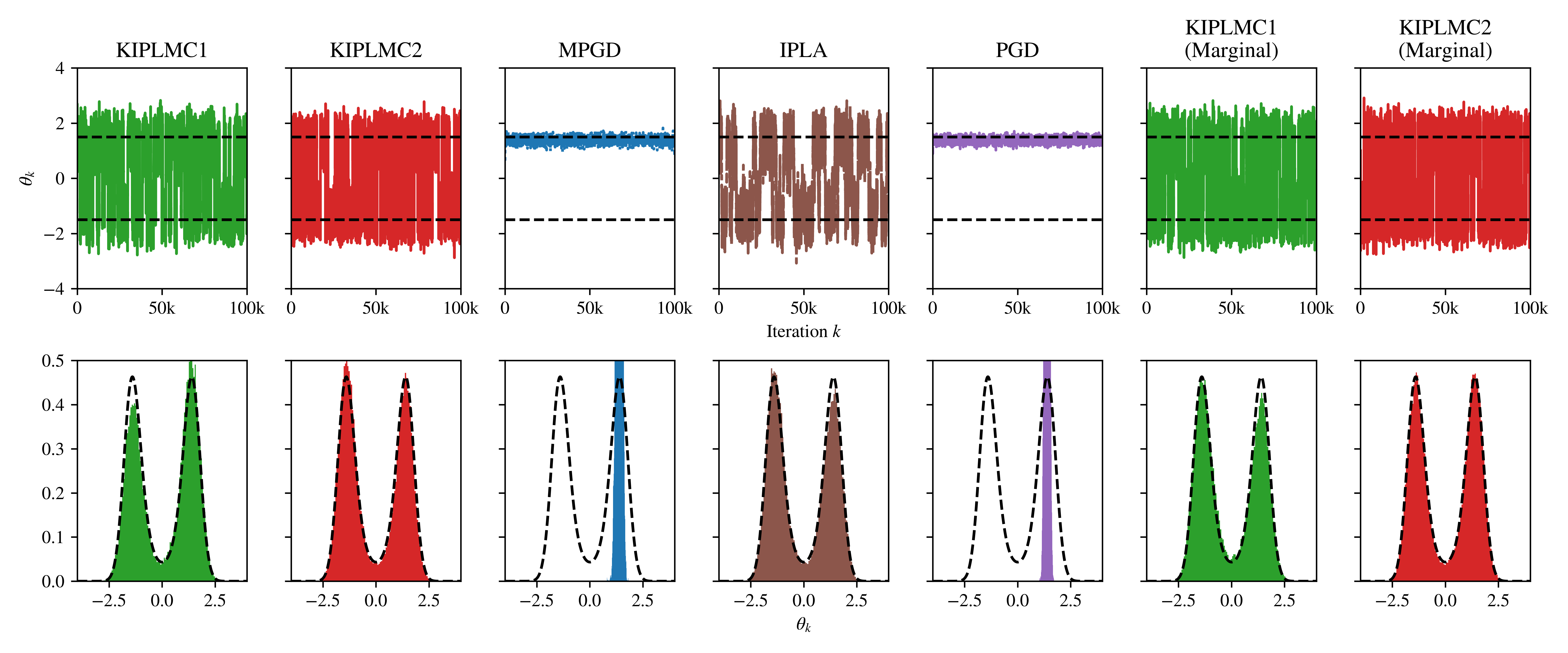}
    \caption{\textbf{Multi-well Experiment.} The trace plots for \gls*{kiplmc1}, \gls*{kiplmc2}, \gls*{mpgd}, \gls*{ipla} and \gls*{pgd} on the multi-well experiment, together with density plots. The black dashed lines in the trace plots indicate the observationally equivalent data-generating parameter values $\theta=\pm1.5$, while the black dashed curves in the density plots show the target $\theta$-density proportional to $p_\theta(y)^N$. $N=10$, whilst the friction and step-size are optimised for each algorithm.}
    \label{fig:multiwell}
\end{figure*}
We include the following multi-well experiment to highlight the advantage of the noised $\theta$ dynamics, the core difference between the proposed methods and \gls*{mpgd} \citep{lim2023momentum}. To do this, we consider a synthetic example with a non-convex negative log-marginal-likelihood objective. Consider the joint probability density over $\mathbb{R}$,
\begin{align*}
    p_\theta(x,y) &= \mathcal{N}(y; x, \sigma^2_y)p_\theta(x),\\
    p_\theta(x) &= \frac{1}{2}\mathcal{N}(x; -\theta,\sigma^2_x) + \frac{1}{2}\mathcal{N}(x;\theta,\sigma_x^2).
\end{align*}
For $|y|>\sqrt{\sigma_x^2+\sigma_y^2}$, the marginal likelihood is bimodal as a function of $\theta$. Indeed, integrating out $x$ gives
\begin{equation*}
p_\theta(y) = \frac{1}{2} \mathcal{N}(y;\theta, \sigma_x^2+\sigma_y^2) + \frac{1}{2}\mathcal{N}(y;-\theta, \sigma_x^2 +\sigma_y^2).
\end{equation*}

In this case, we seek to show that whilst the \gls*{mpgd} algorithm gets stuck in one of the local minima, the \gls*{kiplmc1} algorithm escapes. For the experiment, we generate a single observation from the model with $\theta=1.5$, and we fix $\sigma_x^2=1$ and $\sigma_y^2=0.25$. Since $p_\theta=p_{-\theta}$, the values $\theta=\pm1.5$ are observationally equivalent data-generating parameter values; they need not coincide exactly with the maximisers of the marginal likelihood for the realised observation. In Fig.~\ref{fig:multiwell}, we show the trace plots of $\theta$ together with the associated empirical density plots, compared with these two reference values and the target $\theta$-density proportional to $p_\theta(y)^N$, respectively. For sake of comparison, we also compare the dynamics of \gls*{kiplmc1} and \gls*{kiplmc2} on the marginal dynamics, i.e. we consider \gls*{kipld} to be given simply as
\begin{align*}
\mathrm{d}\bm{\theta}_t &= \V_t^\theta \mathrm{d}t,\\
\mathrm{d}\V^\theta_t &= -\gamma \V_t^\theta\,\mathrm{d}t -\nabla_\theta \kappa(\bm{\theta}_t)\,\mathrm{d}t + \sqrt{\frac{2\gamma}{N}}\,\mathrm{d}\B_t,
\end{align*}
for an $\mathbb{R}^{d_\theta}$-valued Brownian motion $\B_t$. Recalling from Proposition~\ref{prop:stationary} that $\kappa(\theta)=-\log p_\theta(y)$, the invariant $\theta$-density of these marginal dynamics is proportional to $p_\theta(y)^N$. In Fig.~\ref{fig:multiwell}, all methods approach at least one of the two wells, but only the ``noised'' algorithms \gls*{kiplmc1}, \gls*{kiplmc2} and \gls*{ipla} reliably explore both wells.

\begin{figure}[t!]
    \centering
    \includegraphics[width=.5\linewidth]{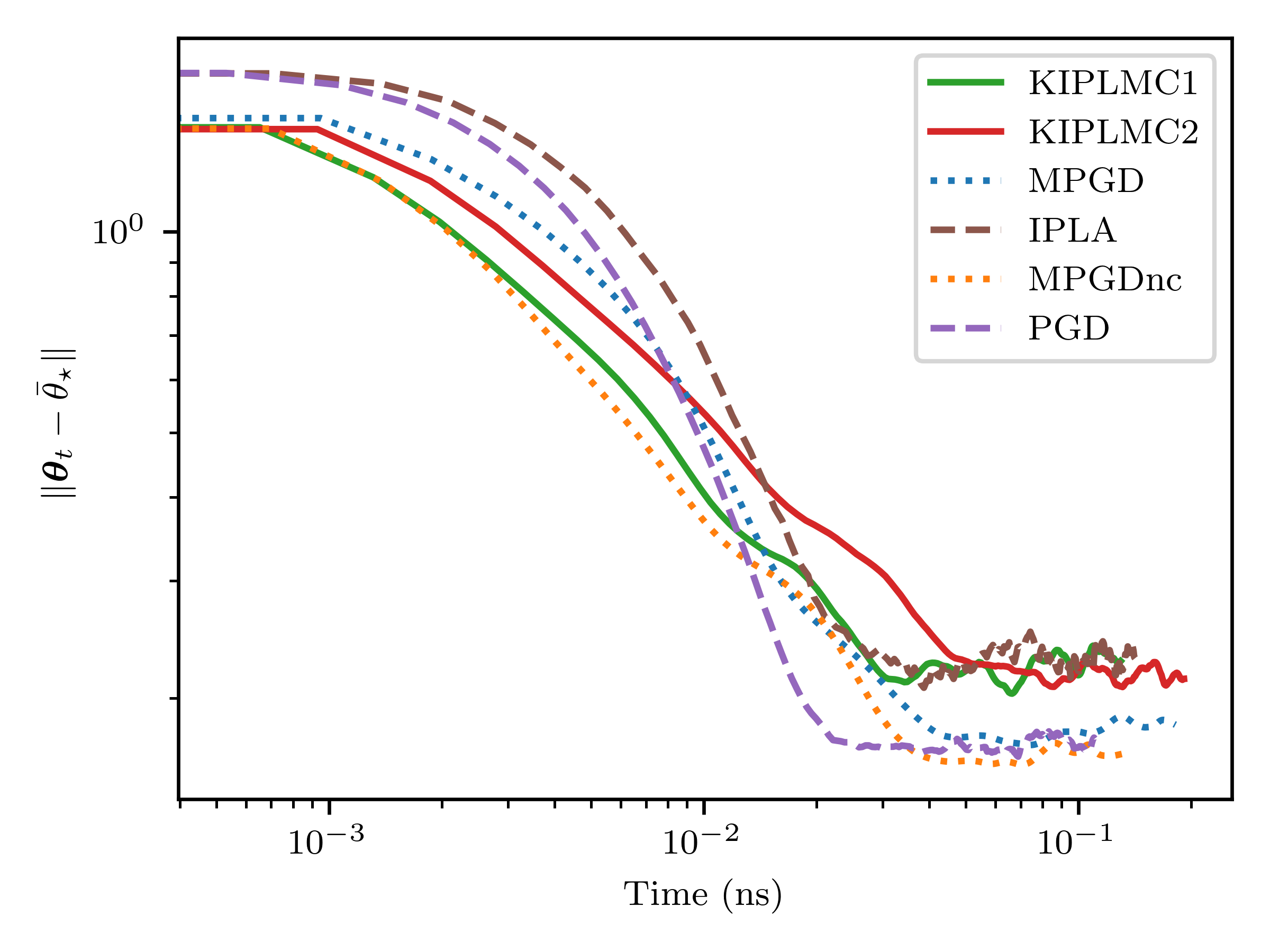}
    \caption{\textbf{Runtime Comparison.} We compare the MSE for the synthetic logistic regression experiment performed in Sec.~\ref{sec:toyex} with the same hyperparameters $d_x=d_\theta=3$, $d_y=100$, $\gamma=2.2$ and $N=100$. The results shown are the average of 100 Monte Carlo simulations.}
    \label{fig:timing}
\end{figure}

\subsection{Computational Cost Comparison}\label{app:compcost}
We further perform an additional experiment on the Sec.~\ref{sec:toyex} example, to compare MSE against computational cost. Indeed, this allows for more thorough comparisons with the methods, as \gls*{kiplmc1} generates correlated noise, \gls*{kiplmc2} performs a gradient correction and \gls*{mpgd} does both. Further all these methods require twice as much space in memory, when compared to \gls*{pgd} and \gls*{ipla}, due to the presence of the momentum variables.

In Fig.~\ref{fig:timing}, we observe the faster initial convergence rate of the momentum based methods, despite their greater computational cost. Further, we observe that the increased cost of noising the $\theta$-dynamics in our proposed methods, is negligible when compared to a second gradient evaluation per iteration, required by \gls*{mpgd} and \gls*{kiplmc2}.

\subsection{Area Between the Curves Calculation (ABC)}\label{sec:ABC}
To compare the behaviour of the different algorithms, we employ a measure that quantifies how accurately and quickly each algorithm converges to the true solution, namely the ABC, as in \citep{lim2023momentum}. Consider $c^k:\mathbb{N}\to \mathbb{R}$ and $c^m:\mathbb{N}\to \mathbb{R}$ to be the accuracy curves of \gls*{kiplmc2} and \gls*{ipla} respectively, i.e. for $\{\theta_n\}_{n=1}^\infty$ generated by the \gls*{kiplmc2} algorithm, $c^k(n)=\|\theta_n -\bar{\theta}^\star\|$. The ABC is given as,
\[\sum_{i=1}^M w(i) (c^m(i)-c^k(i)),\]
where $M$ is the total number of time-steps and $w$ acts as a weighting function $w(i)=2i/(M(M+1))$, matching \citep{kuntz2023particle}. Hence, the ABC computes a signed, weighted area between the error curves of \gls*{ipla} and \gls*{kiplmc2}. In particular, when $c^m$ is dominated by $c^k$, the ABC is negative, whilst positive in the converse case.

\subsection{Error Metrics for the BNN}\label{sec:err}
Recall that the output of our model for latent variable $x=(w,v)$, image features $f$ is,
\begin{equation*}
p(l|f, x) \propto \exp\left(\sum_{j=1}^{40} v_{lj} \tanh\left(\sum_{i=1}^{784} w_{ji} f_i\right)\right).
\end{equation*}
To normalise the outputs of the RHS over the labels we apply softmax to the estimates for all $l$. Denote this as $g(l|f,x)$ and we give the output of the model as
\begin{equation*}
\hat{l}(f|x) = \argmax_{l\in\{0,1\}} g(l|f,x).
\end{equation*}

For the BNN in \ref{sec:bnn} we use two metrics to evaluate the performance of our models, which we evaluate on test set of 200 images and labels $\mathcal{Y}_\text{test}$. In particular we use a relative error metric, measuring the percentage accuracy on the test set,
\begin{equation*}
\frac{1}{|\mathcal{Y}_\text{test}|}\sum_{(f,l)\in\mathcal{Y}_\text{test}} |l - \hat{l}(f|x)|.
\end{equation*}

Following \citep{kuntz2023particle}, our second metric is the log pointwise predictive density (LPPD),
\begin{equation*}
\frac{1}{|\mathcal{Y}_\text{test}|}\sum_{(f,l)\in\mathcal{Y}_\text{test}} \log g(l|f,x).
\end{equation*}
It is the average log probability assigned by the model to the correct test label; larger values indicate better predictive performance.

%% file: ref.bib
@article{saumard2014log,
  title={Log-concavity and strong log-concavity: a review},
  author={Saumard, Adrien and Wellner, Jon A},
  journal={Statistics surveys},
  volume={8},
  pages={45},
  year={2014},
  publisher={NIH Public Access}
}

@article{santambrogio2015optimal,
  title={Optimal transport for applied mathematicians},
  author={Santambrogio, Filippo},
  journal={Birk{\"a}user, NY},
  volume={55},
  number={58-63},
  pages={94},
  year={2015},
  publisher={Springer}
}

@article{akyildiz2024nonasymptotic,
  title={Nonasymptotic analysis of Stochastic Gradient Hamiltonian Monte Carlo under local conditions for nonconvex optimization},
  author={Akyildiz, \"O. Deniz and Sabanis, Sotirios},
  journal={Journal of Machine Learning Research},
  volume={25},
  number={113},
  pages={1--34},
  year={2024}
}

@book{bernardo2009bayesian,
  title={Bayesian theory},
  author={Bernardo, Jos{\'e} M and Smith, Adrian FM},
  volume={405},
  year={2009},
  publisher={John Wiley \& Sons},
  address = {New York}
}

@article{meng1993maximum,
  title={Maximum likelihood estimation via the ECM algorithm: A general framework},
  author={Meng, Xiao-Li and Rubin, Donald B},
  journal={Biometrika},
  volume={80},
  number={2},
  pages={267--278},
  year={1993},
  publisher={Oxford University Press}
}

@article{liu1994ecme,
  title={The ECME algorithm: a simple extension of EM and ECM with faster monotone convergence},
  author={Liu, Chuanhai and Rubin, Donald B},
  journal={Biometrika},
  volume={81},
  number={4},
  pages={633--648},
  year={1994},
  publisher={Oxford University Press}
}

@article{lange1995gradient,
  title={A gradient algorithm locally equivalent to the EM algorithm},
  author={Lange, Kenneth},
  journal={Journal of the Royal Statistical Society: Series B (Methodological)},
  volume={57},
  number={2},
  pages={425--437},
  year={1995},
  publisher={Wiley Online Library}
}

@article{monmarche2021high,
  title={High-dimensional {MCMC} with a standard splitting scheme for the underdamped {L}angevin diffusion.},
  author={Monmarch{\'e}, Pierre},
  journal={Electronic Journal of Statistics},
  volume={15},
  number={2},
  pages={4117--4166},
  year={2021},
  publisher={The Institute of Mathematical Statistics and the Bernoulli Society}
}

@book{billingsley1995probability,
  title={Probability and Measure},
  author={Billingsley, P.},
  isbn={9780471007104},
  lccn={gb95051456},
  series={Wiley Series in Probability and Statistics},
  address={Hoboken},
  year={1995},
  publisher={Wiley}
}

@article{akyildiz2022statistical,
  title={Statistical finite elements via {L}angevin dynamics},
  author={Akyildiz, {\"O}mer Deniz and Duffin, Connor and Sabanis, Sotirios and Girolami, Mark},
  journal={SIAM/ASA Journal on Uncertainty Quantification},
  volume={10},
  number={4},
  pages={1560--1585},
  year={2022},
  publisher={SIAM}
}

@article{akyildiz2025interacting,
author = {Akyildiz, {\"O}. Deniz and Crucinio, Francesca Romana and
Girolami, Mark and Johnston, Tim and Sabanis, Sotirios},
journal = {ESAIM: Probability and Statistics},
title = {Interacting Particle {L}angevin Algorithm for Maximum
Marginal Likelihood Estimation},
year = {2025}
}

@article{johnston2023kinetic,
  title={Kinetic Langevin MCMC Sampling Without Gradient Lipschitz Continuity-the Strongly Convex Case},
  author={Johnston, Tim and Lytras, Iosif and Sabanis, Sotirios},
  journal={Journal of Complexity},
  year={2024},
  publisher={Elsevier}
}

@article{fan2023gradient,
  title={Gradient flows for empirical Bayes in high-dimensional linear models},
  author={Fan, Zhou and Guan, Leying and Shen, Yandi and Wu, Yihong},
  journal={arXiv preprint arXiv:2312.12708},
  year={2023}
}

@article{dempster1977maximum,
  title={Maximum likelihood from incomplete data via the EM algorithm},
  author={Dempster, Arthur P and Laird, Nan M and Rubin, Donald B},
  journal={Journal of the {R}oyal Statistical Society: Series B (Methodological)},
  volume={39},
  number={1},
  pages={1--22},
  year={1977},
  publisher={Wiley Online Library}
}

@inproceedings{lim2023momentum,
  title={Momentum particle maximum likelihood},
  author={Lim, Jen Ning and Kuntz, Juan and Power, Samuel and Johansen, Adam M},
  booktitle={Proceedings of 41st International Conference on Machine Learning (ICML)},
  volume={235},
  year={2024}
}

@article{whiteley2022statistical,
 title={Statistical exploration of the Manifold Hypothesis},
  author={Whiteley, Nick and Gray, Annie and Rubin-Delanchy, Patrick},
  journal={Journal of the Royal Statistical Society: Series B},
  year={2025}
}

@inproceedings{kuntz2023particle,
  title={Particle algorithms for maximum likelihood training of latent variable models},
  author={Kuntz, Juan and Lim, Jen Ning and Johansen, Adam M},
  booktitle={International Conference on Artificial Intelligence and Statistics},
  pages={5134--5180},
  year={2023},
  organization={PMLR}
}

@article{altschuler2023fasterhighaccuracylogconcavesampling, author = {Altschuler, Jason M. and Chewi, Sinho}, title = {Faster High-accuracy Log-concave Sampling via Algorithmic Warm Starts}, year = {2024}, issue_date = {June 2024}, publisher = {Association for Computing Machinery}, address = {New York, NY, USA}, volume = {71}, number = {3}, issn = {0004-5411}, url = {https://doi.org/10.1145/3653446}, doi = {10.1145/3653446}, journal = {J. ACM}, month = jun, articleno = {24}, numpages = {55} }

@article{HOROWITZ1991247,
title = {A generalized guided Monte Carlo algorithm},
journal = {Physics Letters B},
volume = {268},
number = {2},
pages = {247-252},
year = {1991},
issn = {0370-2693},
doi = {https://doi.org/10.1016/0370-2693(91)90812-5},
url = {https://www.sciencedirect.com/science/article/pii/0370269391908125},
author = {Alan M. Horowitz}
}

@InProceedings{coldposterior,
  title = 	 {How Good is the {B}ayes Posterior in Deep Neural Networks Really?},
  author =       {Wenzel, Florian and Roth, Kevin and Veeling, Bastiaan and Swiatkowski, Jakub and Tran, Linh and Mandt, Stephan and Snoek, Jasper and Salimans, Tim and Jenatton, Rodolphe and Nowozin, Sebastian},
  booktitle = 	 {Proceedings of the 37th International Conference on Machine Learning},
  pages = 	 {10248--10259},
  year = 	 {2020},
  editor = 	 {Hal Daumé and Singh, Aarti},
  volume = 	 {119},
  series = 	 {Proceedings of Machine Learning Research},
  month = 	 {July},
  publisher =    {PMLR},
  address = {Online}
}

@article{de2021efficient,
  title={Efficient stochastic optimisation by unadjusted Langevin Monte Carlo},
  author={De Bortoli, Valentin and Durmus, Alain and Pereyra, Marcelo and Vidal, Ana F},
  journal={Statistics and Computing},
  volume={31},
  number={3},
  pages={1--18},
  year={2021},
  publisher={Springer}
}

@article{johnston2024taming,
	author = {Johnston, Tim and Makras, Nikolaos and Sabanis, Sotirios},
	date = {2025/06/13},
	doi = {10.1007/s00245-025-10269-z},
	id = {Johnston2025},
	isbn = {1432-0606},
	journal = {Applied Mathematics \& Optimization},
	number = {3},
	pages = {77},
	title = {Taming the Interacting Particle Langevin Algorithm: The Superlinear case},
	volume = {91},
	year = {2025}
  }

@article{gao2022global,
  title={Global convergence of stochastic gradient hamiltonian monte carlo for nonconvex stochastic optimization: Nonasymptotic performance bounds and momentum-based acceleration},
  author={Gao, Xuefeng and G{\"u}rb{\"u}zbalaban, Mert and Zhu, Lingjiong},
  journal={Operations Research},
  volume={70},
  number={5},
  pages={2931--2947},
  year={2022},
  publisher={INFORMS}
}

@inproceedings{sharrock2024tuning,
  title={Tuning-free maximum likelihood training of latent variable models via coin betting},
  author={Sharrock, Louis and Dodd, Daniel and Nemeth, Christopher},
  booktitle={International Conference on Artificial Intelligence and Statistics},
  pages={1810--1818},
  year={2024},
  organization={PMLR}
}

@article{glyn2024statistical,
  title={Statistical finite elements via interacting particle {L}angevin dynamics},
  author={Glyn-Davies, Alex and Duffin, Connor and Kazlauskaite, Ieva and Girolami, Mark and Akyildiz, {\"O} Deniz},
  journal={SIAM/ASA Journal on Uncertainty Quantification},
  volume={13},
  number={3},
  pages={1200--1227},
  year={2025},
  publisher={SIAM}
}

@article{encinar2025proximal,
  title={Proximal Interacting Particle {L}angevin Algorithms},
  author={Cordero-Encinar, Paula and Crucinio, Francesca R and Akyildiz, O Deniz},
  journal={Uncertainty in Artificial Intelligence (UAI)},
  year={2025}
}

@article{akyildiz2024multiscale,
  title={A Multiscale Perspective on Maximum Marginal Likelihood Estimation},
  author={Akyildiz, \"O. Deniz and Ottobre, Michela and Souttar, Iain},
  journal={Bernoulli (to appear)},
  year={2026}
}

@article{dalalyan2019user,
  title={User-friendly guarantees for the {L}angevin {M}onte {C}arlo with inaccurate gradient},
  author={Dalalyan, Arnak S and Karagulyan, Avetik},
  journal={Stochastic Processes and their Applications},
  volume={129},
  number={12},
  pages={5278--5311},
  year={2019},
  publisher={Elsevier}
}

@article{hwang1980laplace,
  title={Laplace's method revisited: weak convergence of probability measures},
  author={Hwang, Chii-Ruey},
  journal={The Annals of Probability},
  pages={1177--1182},
  year={1980},
  publisher={JSTOR}
}

@article{dalalyan2017theoretical,
  title={Theoretical guarantees for approximate sampling from smooth and log-concave densities},
  author={Dalalyan, Arnak S},
  journal={Journal of the Royal Statistical Society: Series B (Statistical Methodology)},
  volume={79},
  number={3},
  pages={651--676},
  year={2017},
  publisher={Wiley Online Library}
}

@article{zhang2023nonasymptotic,
  title={Nonasymptotic estimates for Stochastic Gradient Langevin Dynamics under local conditions in nonconvex optimization},
  author={Zhang, Ying and Akyildiz, {\"O}. Deniz and Damoulas, Theodoros and Sabanis, Sotirios},
  journal={Applied Mathematics \& Optimization},
  volume={87},
  number={2},
  pages={25},
  year={2023},
  publisher={Springer}
}

@article{dalalyan2018sampling,
  title={On sampling from a log-concave density using kinetic Langevin diffusions},
  author={Dalalyan, Arnak S and Riou-Durand, Lionel},
  journal={Bernoulli},
  volume={26},
  number={3},
  pages={1956--1988},
  year={2020}
}

@inproceedings{cheng2018underdamped,
  title={Underdamped {L}angevin {MCMC}: {A} non-asymptotic analysis},
  author={Cheng, Xiang and Chatterji, Niladri S and Bartlett, Peter L and Jordan, Michael I},
  booktitle={Conference On Learning Theory},
  pages={300--323},
  year={2018}
}

@inproceedings{raginsky2017non,
  title={Non-convex learning via {S}tochastic {G}radient {L}angevin {D}ynamics: a nonasymptotic analysis},
  author={Raginsky, Maxim and Rakhlin, Alexander and Telgarsky, Matus},
  booktitle={Conference on Learning Theory},
  pages={1674--1703},
  year={2017}
}

@article{durmus2017nonasymptotic,
  title={Nonasymptotic convergence analysis for the unadjusted {L}angevin algorithm},
  author={Durmus, Alain and Moulines, Eric},
  journal={The Annals of Applied Probability},
  volume={27},
  number={3},
  pages={1551--1587},
  year={2017},
  publisher={Institute of Mathematical Statistics}
}

@book{oksendal2013stochastic,
  title={Stochastic differential equations: an introduction with applications},
  author={Oksendal, Bernt},
  year={2013},
  edition={6th},
  address={Berlin, Germany},
  publisher={Springer Science \& Business Media}
}

@book{kloeden2012numerical,
  title     = "Numerical solution of {SDE} through computer experiments",
  author    = "Kloeden, Peter Eris and Platen, Eckhard and Schurz, Henri",
  publisher = "Springer",
  series    = "Universitext",
  edition   =  1,
  month     =  dec,
  year      =  1993,
  address   = "Berlin, Germany",
  language  = "en"
}

@article{atchade2017,
  author  = {Yves F. Atchad{{\'e}} and Gersende Fort and Eric Moulines},
  title   = {On Perturbed Proximal Gradient Algorithms},
  journal = {Journal of Machine Learning Research},
  year    = {2017},
  volume  = {18},
  number  = {10},
  pages   = {1--33},
  url     = {http://jmlr.org/papers/v18/15-038.html}
}

@book{Pavliotis_2014, address={New York}, title={Stochastic processes and applications: Diffusion Processes, the Fokker-Planck and langevin equations}, publisher={Springer}, author={Pavliotis, Grigorios A.}, year={2014}}

@article{durmus2019high,
  title={High-dimensional {B}ayesian inference via the unadjusted {L}angevin algorithm},
  author={Durmus, Alain and Moulines, Eric},
  journal={Bernoulli},
  volume={25},
  number={4A},
  pages={2854--2882},
  year={2019},
  publisher={Bernoulli Society for Mathematical Statistics and Probability}
}

@article{Nesterovacc,
author = {Yi-An Ma and Niladri S. Chatterji and Xiang Cheng and Nicolas Flammarion and Peter L. Bartlett and Michael I. Jordan},
title = {{Is there an analog of Nesterov acceleration for gradient-based MCMC?}},
volume = {27},
journal = {Bernoulli},
number = {3},
publisher = {Bernoulli Society for Mathematical Statistics and Probability},
pages = {1942 -- 1992},
year = {2021}
}

@article{blei2003latent,
  title={Latent {Dirichlet} allocation},
  author={Blei, David M and Ng, Andrew Y and Jordan, Michael I},
  journal={Journal of Machine Learning Research},
  volume={3},
  number={Jan},
  pages={993--1022},
  year={2003}
}

@article{smaragdis2006probabilistic,
  title={A probabilistic latent variable model for acoustic modeling},
  author={Smaragdis, Paris and Raj, Bhiksha and Shashanka, Madhusudana},
  journal={Advances in Models for Acoustic Processing Workshop, NIPS},
  volume={148},
  pages={8--1},
  year={2006}
}

@article{hoff2002latent,
  title={Latent space approaches to social network analysis},
  author={Hoff, Peter D and Raftery, Adrian E and Handcock, Mark S},
  journal={Journal of the American Statistical association},
  volume={97},
  number={460},
  pages={1090--1098},
  year={2002},
  publisher={Taylor \& Francis}
}

@article{wei1990monte,
  title={A {Monte Carlo implementation of the EM} algorithm and the poor man's data augmentation algorithms},
  author={Wei, Greg CG and Tanner, Martin A},
  journal={Journal of the American Statistical Association},
  volume={85},
  number={411},
  pages={699--704},
  year={1990},
  publisher={Taylor \& Francis}
}

@article{celeux1985sem,
  title={The {SEM} algorithm: a probabilistic teacher algorithm derived from the EM algorithm for the mixture problem},
  author={Celeux, Gilles},
  journal={Computational Statistics Quarterly},
  volume={2},
  pages={73--82},
  year={1985}
}

@article{celeux1992stochastic,
  title={A stochastic approximation type {EM} algorithm for the mixture problem},
  author={Celeux, Gilles and Diebolt, Jean},
  journal={Stochastics: An International Journal of Probability and Stochastic Processes},
  volume={41},
  number={1-2},
  pages={119--134},
  year={1992},
  publisher={Taylor \& Francis}
}

@article{chan1995monte,
  title={Monte {Carlo EM} estimation for time series models involving counts},
  author={Chan, KS and Ledolter, Johannes},
  journal={Journal of the American Statistical Association},
  volume={90},
  number={429},
  pages={242--252},
  year={1995},
  publisher={Taylor \& Francis}
}

@article{sherman1999conditions,
  title={Conditions for convergence of {Monte Carlo EM} sequences with an application to product diffusion modeling},
  author={Sherman, Robert P and Ho, Yu-Yun K and Dalal, Siddhartha R},
  journal={The Econometrics Journal},
  volume={2},
  number={2},
  pages={248--267},
  year={1999},
  publisher={Oxford University Press Oxford, UK}
}

@article{booth1999maximizing,
  title={Maximizing generalized linear mixed model likelihoods with an automated {Monte Carlo EM} algorithm},
  author={Booth, James G and Hobert, James P},
  journal={Journal of the Royal Statistical Society: Series B (Statistical Methodology)},
  volume={61},
  number={1},
  pages={265--285},
  year={1999},
  publisher={Wiley Online Library}
}

@article{cappe1999simulation,
  title={Simulation-based methods for blind maximum-likelihood filter identification},
  author={Capp{\'e}, Olivier and Doucet, Arnaud and Lavielle, Marc and Moulines, Eric},
  journal={Signal Processing},
  volume={73},
  number={1-2},
  pages={3--25},
  year={1999},
  publisher={Elsevier}
}

@incollection{diebolt1995stochastic,
  title={A stochastic {EM} algorithm for approximating the maximum likelihood estimate},
  author={Diebolt, J and Ip, E HS},
  year={1996},
  booktitle = {Markov Chain Monte Carlo in Practice},
  editor = {W. R. Gilks and S. T. Richardson and D. J. Spiegelhalter},
  publisher={CRC Publishers},
  address={Boca Raton}
}

@article{wu1983convergence,
  title={On the convergence properties of the EM algorithm},
  author={Wu, CF Jeff},
  journal={The Annals of statistics},
  pages={95--103},
  year={1983},
  publisher={JSTOR}
}

@article{caprio2024error,
  title={Error bounds for particle gradient descent, and extensions of the log-Sobolev and Talagrand inequalities},
  author={Caprio, Rocco and Kuntz, Juan and Power, Samuel and Johansen, Adam M},
  journal={Journal of Machine Learning Research},
  volume={26},
  number={103},
  pages={1--38},
  year={2025}
}
